\newcommand{\short}[1]{#1}     % LEAVE UNCOMMENTED
\newcommand{\extended}[1]{}    % LEAVE UNCOMMENTED
\renewcommand{\short}[1]{}      % UNCOMMENT FOR THE EXTEDED VERSION. COMMENT OUT FOR THE SHORT
\renewcommand{\extended}[1]{#1}  % DO WHAT YOU PLEASE

\documentclass{llncs}

\usepackage[cmex10]{amsmath}
\usepackage{amssymb}% ,amsthm}

\usepackage{tikz}
\usepackage{xspace}
\usepackage{float}
\usepackage{wojtek11}

\newcommand{\Purge}{\mathit{Purge}}
\newcommand{\High}{\ensuremath{H}\xspace}
\newcommand{\Environ}{\High}

\newcommand{\low}{{u_l}}
\newcommand{\Low}{\ensuremath{L}\xspace}
\newcommand{\NI}{\textit{NI}}

\newcommand{\NItxt}{noninterference }
\newcommand{\safest}{\mathbb{S}}
\newcommand{\reachst}{\mathbb{T}}

\definecolor{darkblue}{rgb}{0.0, 0.0, 0.55}
\newcommand{\WJ}[1]{{\color{darkgreen}[WJ: #1]}}
\renewcommand{\WJ}[1]{}

\renewcommand{\Agt}{\mathfrak{U}}
\renewcommand{\Act}{\mathfrak{A}}
\renewcommand{\norecl}{\mathfrak{pos}}
\renewcommand{\recl}{\mathfrak{Rec}}

\newcommand{\modl}{\mathcal{M}}

\newcommand{\uwNIL}{\;{\sim}_{{NI}_\Low}}
\newcommand{\ustar}{\mathcal{U}_M^*}
\newcommand{\rstar}{R_M^*}
\newcommand{\sherr}{s_\mathit{HErr}}
\newcommand{\slerr}{s_\mathit{LErr}}
\newcommand{\undef}{\mathit{undef}}
\newcommand{\Mname}{\mathit{setMName}}
\newcommand{\Gname}{\mathit{setGName}}
\newcommand{\chkWeb}{\mathit{chkWeb}}
\renewcommand{\act}{\mathit{act}}
\newcommand{\obsL}[1]{{\color{darkgrey}\mathit{#1}}}
\newcommand{\obsH}[1]{{\color{darkgrey}\mathit{#1}}}

\title{Information Security as Strategic (In)effectivity}
\author{}
\institute{}

\author{ Wojciech Jamroga\inst{1} \and Masoud Tabatabaei\inst{2} \WJ{Peter Y. Ryan?}}
\institute{
  Institute of Computer Science, Polish Academy of Sciences
\and
  Interdisciplinary Centre for Security and Trust, University of Luxembourg
\email{w.jamroga@ipipan.waw.pl,\quad masoud.tabatabaei@uni.lu}
}

\begin{document}
\pagestyle{plain}
\maketitle

\begin{abstract}
Security of information flow is commonly understood as preventing \emph{any} information leakage, regardless of how grave or harmless consequences the leakage can have.
\extended{
  % **********
  Even in models where each piece of information is classified as either sensitive or insensitive, the classification is ``hardwired'' and given as a parameter of the analysis, rather than derived from more fundamental features of the system.
  % **********
}%end extended
In this work, we suggest that information security is not a goal in itself, but rather a means of preventing potential attackers from compromising the correct behavior of the system. To formalize this, we first show how two information flows can be compared by looking at the adversary's ability to harm the system. Then, we propose that the information flow in a system is \emph{effectively information-secure} if it does not allow for more harm than its idealized variant based on the classical notion of noninterference.
\end{abstract}

%\todo{cite all the relevant contexts (effectivity functions)}

\section{Introduction}\label{sec:intro}

\extended{ %begin extended
  % **********
  Information plays multiple roles in interaction between agents (be it humans or artificial entities, e.g., software agents). First, it can be the commodity that the agents compete for; in that case, it often defines the outcome of the ``interaction game''. Key exchange protocols are a good example here, as the involved honest parties strive to learn the key of the other agent while at the same time preventing any information leak to the intruder.
  Secondly, information can define the semantic content of an action: typically, most actions specified in a security protocol consist in transmitting or processing some information. Thirdly, information can be a resource that enables actions and influences the outcome of the game. This is because agents need information to construct and execute plans that can be used to achieve their goals.

  Most approaches to information flow security adopt the first perspective. That is, information defines the ultimate goal of the interaction.
  % **********
}%end extended
\short{
  In most approaches to information flow security, information defines the ultimate goal of the interaction between agents.
}% end-short
Classical information security properties specify \emph{what} information must not leak, and \emph{how} it could possibly leak (i.e., what channels of information leakage are considered), but they do not give account of \emph{why} the information should not leak to the intruder. For example, the property of \emph{noninterference}~\cite{goguen1982security} assumes that the ``low clearance'' users cannot learn anything about the activities of the ``high clearance'' users. In order to violate this, the ``low'' users can try to analyse their observations and/or execute a sequence of explorative actions of their own. \emph{Nondeducibility on strategies}~\cite{Wittbold90information} makes the same assumption about \emph{what} should not leak, but takes also into account covert channels that some ``high'' users can use to send signals to the ``low'' agents according to a previously agreed code.
\emph{Anonymity} in voting~\cite{Chaum81untreaceable,Fujioka92voting} captures that an observer cannot learn what candidate a particular has voted for by looking at the voter's behavior, scanning the web bulletin board, coercing the voter to hand in the vote receipt, etc.

As a consequence, the classical properties of information security can only distinguish between relevant and irrelevant information leaks if the distinction is given explicitly as a parameter, e.g., by classifying available actions into sensitive and insensitive~\cite{goguen1982security}. However, it is usually hard (if not impossible) to obtain such a distinction based on the internal characteristics of the actions. We illustrate the point below by means of a real-life example.

\extended{\subsection{Motivating Example: Phone Banking}\label{sec:motivating}}
\short{\begin{example}\label{ex:motivating}}
In some phone banking services, the maiden name of the user's mother is used as a part of authentication\extended{, e.g., to change the settings of the account. That is, the user is typically asked to spell her name, birth-date, current address, and her mother's maiden name in order to change the credit limit in the account, block/unblock ATM use in specified geographical areas, and so on.
Note that information about one's birth-date and address is fairly easy to obtain in public directories and/or repositories kept and marketed by various web services that require the data for registration. So, the mother's maiden name plays the role of a "strong test of identity" in this scenario}.\footnote{
  This is a real-life example from the authors' personal experience\extended{ with BNP Paribas in one of EU countries}.
  For similar security questions, used by various phone or web services, cf.~e.g.~\cite{Levin08honeymoon}.
  %We could also use the name of one's first pet in the example, and it wouldn't change the argument.
  }
Consider now a user posting an essay about some ancestor of hers on her blog, mentioning also the name of the ancestor.
If the essay is about the user's mother, it reveals potentially dangerous information.
\extended{
  % **********
  This is, among other things, because an intruder can use the information to: (1) access the phone banking service, (2) authenticate impersonating the user, (3) change (in the user's profile) the telephone number used for web banking password recovery and sms authentication of web banking transactions, (4) change the web banking password of the user, and finally (5) log in and transfer money from the user's account.

  % **********
}%end extended
On the other hand, if the post is about some other member of the user's family (father, grandmother, paternal grandfather, etc.) revealing the name of the person is probably harmless.
Note that it is impossible to distinguish between the two pieces of information (say, the mother's maiden name vs. the grandmother's maiden name) based on their internal features.
\extended{
  % **********
  Both have the same syntactic structure of a single word (i.e., a string of characters with no blank spaces) and the same semantic content (a family name of a person; more precisely, the family name of the person at birth).
  % **********
}%end extended
The only difference lies in the context: the first kind of information is used in some important social procedures, while the second one is not.
\short{\qed\end{example}}

\extended{\subsection{Information as Strategic Resource}}

In this paper, we claim that a broader perspective is needed to appropriately model and analyse such scenarios. Agents compete for information not for its own sake, but for reasons that go beyond purely epistemic advantages.
\extended{
  % **********
  An intruder may want to know the password of a PayPal user in order to impersonate the user and steal some \emph{real} money by making a payment to his own benefit (possibly via an account of a suitable ``mule''). An industry player may need encryption keys used in internal communication between employees of its main competitor in order to find out about the competitor's current business strategy.
  A political activist needs the ability to learn the value of someone's vote in order to effectively coerce that person into voting for the candidate that the activist is rallying for.
  Thus, in most security scenarios, information is a resource rather than a commodity.
  % **********
}%end extended
More precisely, information is a commodity that the players compete for in an ``information security game'' but the game is played in the context of a ``real'' game where information is only a resource, enabling (some) players to achieve their non-epistemic goals. As players obtain new information, their uncertainty is reduced, and they increase their ability to choose a good strategy in the real game.

What would a \emph{significant information leak} be in this view? To answer the question, we draw inspiration from the concept of the \emph{value of information}\extended{ from decision theory~\cite{Howard66infovalue}}: a piece of information is worth as much as it increases the expected payoff of the player.
Similarly, an information leak is significant if it increases the ability of the attacker to construct a damaging attack strategy in the real game.

\extended{\subsection{Main Idea and Contribution of the Paper}}
\short{\para{Contribution of the Paper.}}
\extended{
  % **********
  The main idea behind this paper can be summarized as follows.
  We consider three research questions:
  \begin{itemize}
  \item How can we evaluate the ability of an adversary to harm the goal of the system?
  \item How can we compare two systems with regard to the ability of attackers to harm the goal of the system in them?
  \item How can we know whether the ability (or inability) of the attacker to harm the goal of the system is because of some leakage of information to the attacker or not?
  \end{itemize}

  The paper is structured to discuss these questions in order.
  % **********
}%end extended
First, we use the concept of \emph{surely winning strategies} from game theory to analyze the adversary's strategic ability to disrupt the correct behavior of the system.
\extended{
  % **********
  This can be a functionality property, a security property, or a combination of the two kinds. Also, it can arise from a goal of the ``high clearance'' agents or from an objective assigned to the system by its designers and/or owners.
  Preventing the attacker from having a winning attack strategy is what the designer of the system may want to achieve.
  % **********
}%end extended
We will see the \emph{effective security} of the system as the attacker's \emph{in}ability to come up with such a strategy.

Secondly, we use the notion of \emph{effective security} for comparing two systems by looking at the strategic ability of an adversary to harm the goal of the system.

Thirdly, a successful attack strategy can exist due to flawed design of either the control flow or the information flow in the system. Here, we are interested in the latter. That is, we want to distinguish between vulnerabilities coming from the control vs.~the information flow, and single out systems where redesigning the flow of information alone can make the system more secure. To this end, we define the noninterferent idealized variant of the system, which has the same control flow as the original system, but with the information reduced so that the system satisfies noninterference. Then, we define the system to be \emph{effectively information-secure} if it is as good as its noninterfering idealized variant.
As the main technical result, we show that the concept is well defined, i.e., the maximal noninterferent variant exists for every state-transition model.
\extended{
  % **********

  We begin by presenting the preliminary concepts (models of interaction, noninterference, strategies and their outcomes) in Section~\ref{sec:preliminaries}. In Section~\ref{sec:effsec}, we define the generic concept of effective security. In Section~\ref{sec:effinfosec}, we look specifically at the security of information flow, and show how it can be defined based on the relation between the attacker's observational capabilities and his ability to compromise the goals of the system. In Section~\ref{sec:non-total-models} we extend our results to models that are not total on input. Finally, we summarize the work in Section~\ref{sec:conclusions}.
  % **********
}%end extended

\section{Related Work}\label{sec:relatedwork}

Various formalizations of information flow security have been proposed and studied.
The classical concept here is \emph{noninterference}~\cite{goguen1982security} and its variations: \emph{nondeducibility}~\cite{sutherland1986model}, \emph{noninference}~\cite{o1990calculus}, \emph{restrictiveness}~\cite{mccullough1988noninterference}, \emph{nondeducibility on strategies}~\cite{Wittbold90information}, and \emph{strategic noninterference}~\cite{Jamroga15strat-ni}.
\extended{
  % **********
  Although noninterference was originally introduced for finite transition systems, it was later redefined, generalized, and extended in the framework of process algebras~\cite{allen1991comparison,roscoe1994non,roscoe1995csp,roscoe1999intransitive,ryan2001process}.
  Noninterference and its variants have been studied from different perspectives. Some works dealt with composability of noninterference~\cite{mccullough1988noninterference,zakinthinos1995composability,seehusen2009information}. Another group of papers studied the properties of intransitive noninterference~\cite{roscoe1999intransitive,backes2003intransitive,van2007indeed,engelhardt2012intransitive} which is important in systems with downgraders.
  % **********
}%end extended
Probabilistic noninterference and quantitative noninterference have been investigated, e.g., in~\cite{gray1990probabilistic,Wittbold90information,mciver2003probabilistic,pierro2004approximate,li2005downgrading,smith2009foundations}.
All the above concepts assume that the information flow in the system is secure only when no information ever flows from High to Low players.
In this paper, we want to discard irrelevant information leaks, and only look at the significant ones (in the sense that the leaking information can be used to construct an attack on a higher-order correctness property).

The problem of how to weaken noninterference to successfully capture security guarantees for real systems has been also extensively studied.
Most notably, postulates and policies for \emph{declassification} (called also \emph{information release}) were studied, cf.~\cite{SabelfeldS05declassification} for an introduction. This submission can be viewed as an attempt to determine \emph{what information is acceptable to declassify}. In this sense, our results can useful in proposing new declassification policies and evaluating existing ones. We note, however, that the existing work on declassification are mainly concerned by the question \emph{what} information can be released, \emph{when}, \emph{where}, and by \emph{whom}. In contrast, we propose an argument for \emph{why} it can be released.
Moreover, declassification is typically about intentional release of information, whereas we do not distinguish between intentional and accidental information flow.
Finally, the research on declassification assumes that security is defined by some given ``secrets'' to be protected. In our approach, no information is intrinsically secret, but the information flow is harmful if it enables the attacker to gain more strategic ability against the goals of the system.

Parameterized noninterference~\cite{Giacobazzi04abstract-noninterf} can be seen as a theoretical counterpart of declassification, where security of information flow is parameterized by the analytic capabilities of the attacker. Again, that research does not answer why some information must be kept secret while some other needs not, and in particular it does not take strategic power of the attacker into account.

Economic and strategic analysis of security properties is a growing field in general,
\extended{cf.~e.g.~\cite{Anderson07incentives,Moore11economics,Dodis07cryptoGT,Buldas07evoting,Yin10stackelberg,Korzhyk11stackelberg}.}
\short{cf.~\cite{Moore11economics} for an introduction.}
A number of papers have applied game-theoretic concepts to define the security of information flow~\cite{Malacaria99nondet-games,Hankin02flow,Harris13capabilities,Dimovski14noninterf-game,Fielder14GT-infosec,Jamroga15strat-ni}.
However, most of those papers~\cite{Malacaria99nondet-games,Hankin02flow,Harris13capabilities,Dimovski14noninterf-game} use games only in a narrow mathematical sense to provide a proof system (called the \emph{game semantics}) for deciding security properties.
We are aware of only a handful of papers that investigate the impact of participants' incentives and available strategies on the security of information flow.
In~\cite{AcquistiG04privacy,Grossklags08secureOrInsure}, economic interpretations of privacy-preserving behavior are proposed.
\cite{Fielder14GT-infosec} uses game-theoretic solution concepts (in particular, Nash equilibrium) to prescribe the optimal defense strategy against attacks on information security. In contrast, our approach is analytic rather than prescriptive, as we do not propose how to manage information security. Moreover, in our view, privacy is not the goal but rather the means to achieve some higher-level objectives.
Finally, \cite{Jamroga15strat-ni}~proposes a weaker variant of noninterference by allowing the High players to select an appropriate strategy, while here we look at the potential damage inflicted by adverse strategies of the Low users.

Our idea of looking at the unique most precise non-interfering variant of the system is related on the technical level to~\cite{Giacobazzi04abstract-noninterf}. There, attackers displaying different analytical capabilities are defined by abstract interpretation, which leads to a lattice of noninterference variants with various strength. Attackers with weakened observational powers were also studied in~\cite{Zdancewic03observational-det}.

\section{Preliminaries}\label{sec:preliminaries}

\extended{
  % **********
  We begin by presenting the main ingredients that we are going to use in our proposal. First, we introduce simple models of interaction that slightly extend the classical approach of Goguen and Meseguer. Then we recall Goguen and Meseguer's definition of noninterference that captures the property of secure information flow from the ``insider'' agents to the ``outsiders''. Finally, we present some basic concepts from game theory (strategies, winning strategies) and theory of temporal specification (temporal goals).
  % **********
}%end extended

\subsection{Simple Models of Interaction}\label{sec:models}

Since we build our proposal around the standard notion of noninterference by Goguen and Meseguer~\cite{goguen1982security}, we will use similar models to represent interaction between actions of different agents.
The \emph{system} is modeled by a multi-agent asynchronous transition network
$M = \tuple{\States,s_0,\Agt,\Act,Obs,obs,do}$
where:
$\States$ is the set of \emph{states},
$s_0$ is the initial state,
$\Agt$ is the set of \emph{agents} (or \emph{users}),
$\Act$ is the set of \emph{actions}\extended{ (or \emph{commands})},
$Obs$ is the set of possible \emph{observations} (or \emph{outputs});
$obs : \States \times \Agt \then Obs$ is the observation function.
$do : \States \times \Agt \times \Act \then \States$ is the transition function that specifies the (deterministic) outcome $do(s,u,a)$ of action $a$ executed by user $u$ in state $s$.
%By $do(s,u,a)=\undef$ we denote that action $a$ is unavailable to user $u$ in state $s$;
%additionally, we define $\act(s,u)=\set{a\in\Act \mid do(s,u,a)\neq\undef}$ as the set of actions available to $u$ in $s$.
%The following property is required to hold: For any $u\in\Agt, s_1, s_2 \in\States$, if $obs(s_1,u)=obs(s_2,u)$ then $\act(s_1,u)=\act(s_2,u)$.
%That is, we assume that players are aware of their available actions, and hence can distinguish states with different repertoires of choices.\footnote{
%The models of Goguen and Meseguer~\cite{goguen1982security} additionally assume that the transition function is \emph{total on input}, i.e., every action is available to each player at every state. Thus, they can be seen as the subclass of our models where the transition function $d$ is total. }
%
We will sometimes write $[s]_u$ instead of $obs(s,u)$. Also, we will call a pair \textit{(user, action)} a \emph{personalized action}.
We construct the multi-step transition function $exec : \States \times (\Agt\times\Act)^* \then \States$ so that,
for a finite string $\alpha\in (\Agt\times\Act)^*$ of personalized actions, $exec(s,\alpha)$ denotes the state resulting from execution of $\alpha$ from $s$ on. %If at some point in executing $\alpha$ from $s$ the subsequent personalized action is not available then $exec(s,\alpha) = \undef$.
We may sometimes write $s\xrightarrow{\alpha} t$ instead of $exec(s,\alpha) = t$, and $exec(\alpha)$ instead of $exec(s_0,\alpha)$.
\extended{
  % **********
  The way models are constructed is illustrated by the following example.

  \begin{figure*}[!t]\centering
  \hspace*{-40pt}
  \begin{tikzpicture}[thick,scale=0.7, every node/.style={transform shape}]
  \input{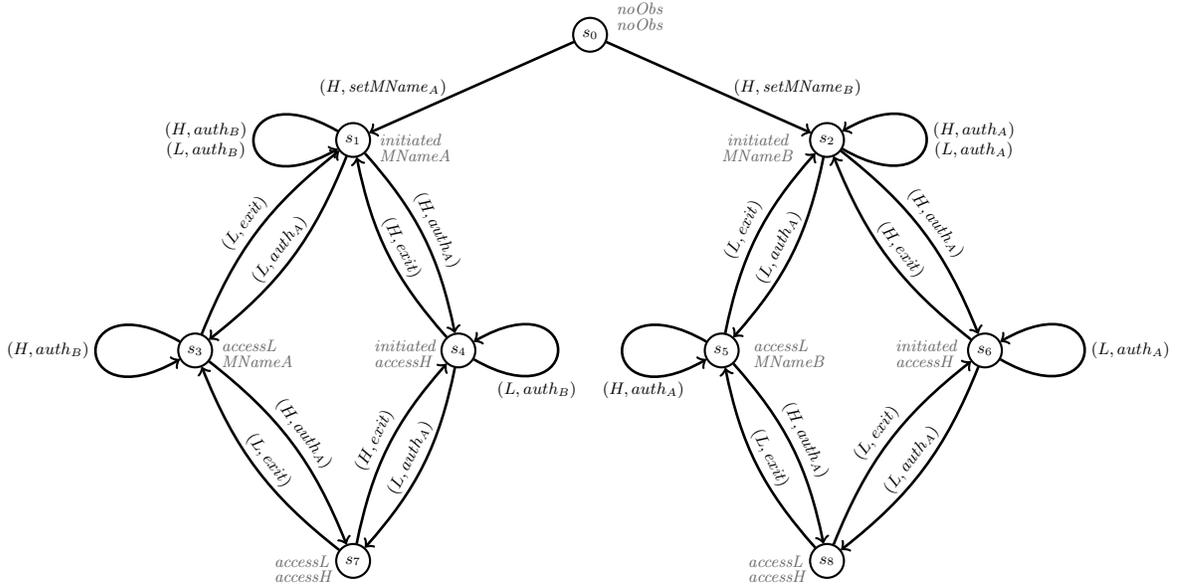}
  \end{tikzpicture}
  \caption{A simple model for the phone banking example}
  \label{fig:banking1}
  \end{figure*}

  \begin{example}\label{ex:banking1}
  Consider a simplified version of the phone banking scenario from Section~\ref{sec:motivating}, where a client can access his/her account by correctly giving the maiden name of his/her mother. Figure~\ref{fig:banking1} presents the simplest possible transition network for the scenario. Labels on transitions show the personalized actions resulting in the transition, and the observations of users in each state are shown beside the state. There are two users: \High who has an account in the bank, and \Low who may try to impersonate \High. At the initial state, \High enters her mother's maiden name when setting up her profile at the bank. To keep the graph simple, we include only two possibilities: action $\Mname_A$ fixes the name as ``A'', whereas action $\Mname_B$ sets the entry to ``B''. Clearly, \High can observe which value she entered (different observations in states $s_1,s_2$). Moreover, \Low cannot observe that, as \Low's observations in $s_1,s_2$ are the same.

  Any user can access \High's bank account by mentioning the name correctly. So, depending on the value that has been entered, doing either $auth_A$ or $auth_B$ will grant the user with access to the bank account. Moreover, the user who authenticated successfully can observe it, and the other user cannot even see that the authentication took place. The user who has successfully logged in to the account, can log out by executing the action $exit$.
  To make the graph simpler, we also assume that giving the wrong name has no effect on the state of the system.
  \end{example}
  % **********
}%end extended

Three remarks are in order.
First, Goguen and Meseguer's models define agents' observations based on states only, whereas it is often convenient to also model the information flow due to observing each others' actions.
Secondly, the models are fully asynchronous in the sense that if each user ``submits'' a sequence of actions to be executed then every interleaving of the submitted sequences can occur as the resulting behavior of the system. No synchronization is possible.
Thirdly, the models are ``total on input'' (each action label is available to every user at every state), and hence no synchronization mechanism can be encoded via availability of actions.
Especially the last two features imply that models of Goguen and Meseguer allow for representation of a very limited class of systems.

\extended{
  % **********
  More expressive classes of models include various kinds of transition systems~\cite{Winskel95concurrency}, concurrent programs~\cite{Kupferman00automata-mcheck}, interpreted systems~\cite{Fagin95knowledge}, reactive modules~\cite{Alur99reactive}, multi-agent transition networks (a.k.a. concurrent game structures)~\cite{Alur02ATL}, and many more.

  \WJ{point out that the structures are a special kind of temporal-epistemic Kripke structures}
  \WJ{note how transition networks obtained by operational semantics of CCS/CSP process expressions fit into this framework}
  % **********
}
We start by using the purely asynchronous models of Goguen and Meseguer. Then, in Section~\ref{sec:non-total-models}, we extend our results to a broader class of models by allowing partial transition functions.
%This is why we have extended the models to allow for partial transition functions.
%However, there is a price to pay: we cannot use the characterization results for information security from~\cite{goguen1984unwinding,rushby1992noninterference,Meyden10comparison}. Instead, we will have to ``upgrade'' the results to the broader class of models.

\subsection{Noninterference}\label{sec:noninterference}

We now recall the standard notion of noninterference from~\cite{goguen1982security}.
%, and extend it to models that are not total on input.
Let $U\subseteq\Agt$ and $\alpha\in(\Agt\times\Act)^*$.
By $\Purge_U(\alpha)$ we mean the subsequence of $\alpha$ obtained by eliminating all the pairs $(u,a)$ with $u\in U$.
%Also, by $\Purge_A(\alpha)$ we mean the subsequence of $\alpha$ obtained by eliminating all the pairs $(u,a)$ with $a\in A$.
%Finally, $\Purge_{U,A}(\alpha)$ denotes the subsequence of $\alpha$ obtained by eliminating all the pairs $(u,a)$ with $u\in U$ and $a\in A$.

\begin{definition}[Noninterference~\cite{goguen1982security}]\label{def:noninterference}
Let $M$ be a transition network with sets of ``high clearance'' agents $\High$ and ``low clearance'' agents $\Low$, such that $\High\cap\Low=\emptyset, \High\cup\Low=\Agt$.
We say that \emph{$\High$ is non-interfering with $\Low$} iff for all $\alpha\in(\Agt\times\Act)^*$ and all $\low\in \Low$, $[exec(\alpha)]_{\low} = [exec(\Purge_\High(\alpha))]_{\low}$. We denote the property by $\NI_M(\High,\Low)$.
%\footnote{  The original definition of noninterference from~\cite{goguen1982security} omits the condition ``if $exec(\alpha)\neq \undef$''. Clearly, both definitions agree in the class of models considered in~\cite{goguen1982security}, i.e., models that are total on input. This is because, in those models, all sequences of personalized actions can be executed from any state in the system. }
\end{definition}

Thus, $\NI_M(\High,\Low)$ expresses that \Low can neither observe nor deduce what actions of \High have been executed
\extended{ -- in fact, they have no clue whether \emph{any} \High's action was executed at all.}

\begin{figure}[!t]%\centering
\hspace*{-30pt}
  \begin{tikzpicture}[thick,scale=0.6, every node/.style={transform shape}]
  \input{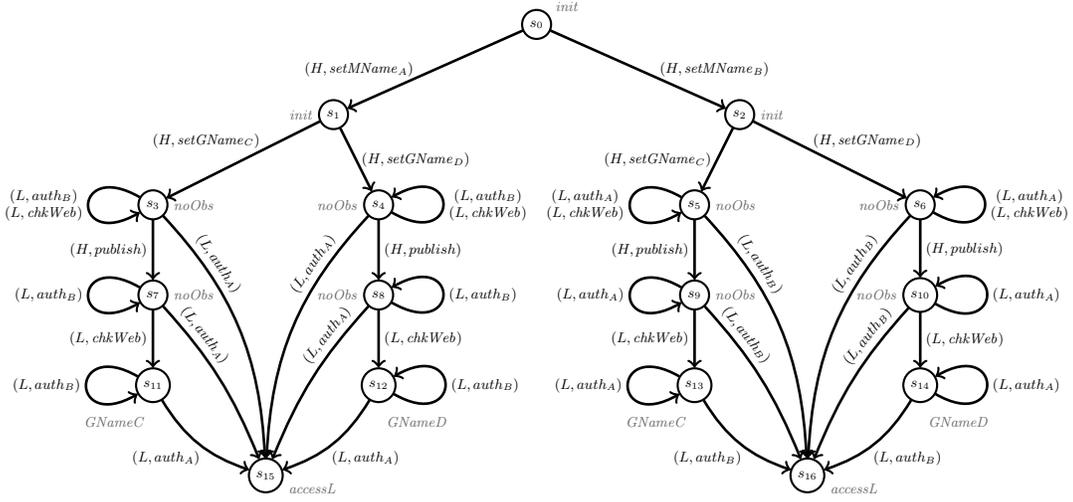}
  \end{tikzpicture}
\caption{Transition network $M_a$ in which the High player publishes her grandmother's maiden name on her blog. Only the observations of \Low are shown}
\label{fig:banking2a}
\end{figure}

\begin{figure}[!t]%\centering
\hspace*{-30pt}
  \begin{tikzpicture}[thick,scale=0.6, every node/.style={transform shape}]
  \input{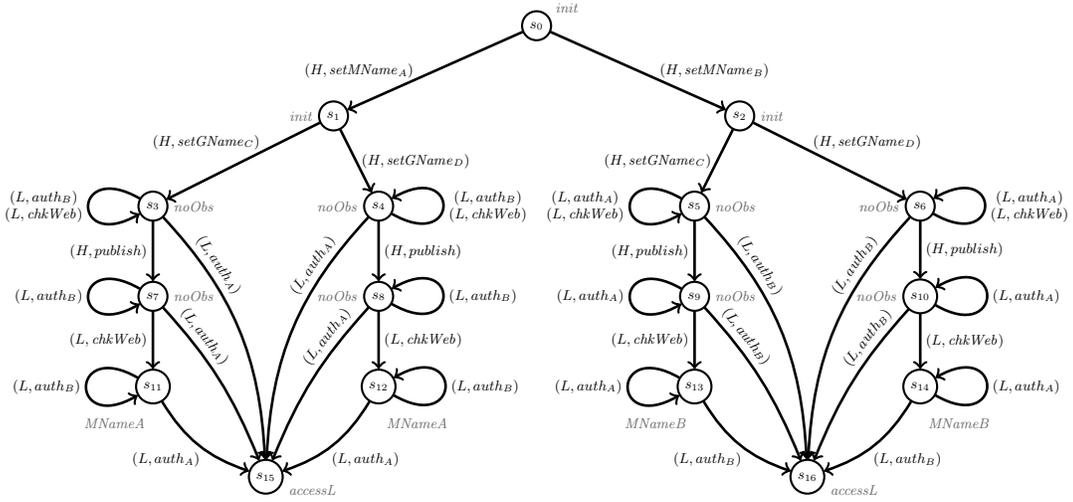}
  \end{tikzpicture}
\caption{Transition network $M_b$ in which \High publishes her mother's maiden name}
\label{fig:banking2b}
\end{figure}

%\caption{Two alternative blog \& phone banking scenarios\extended{. The difference between the two variants is that in (a), after \Low executes $chkWeb$, \Low's observation does not help her to have a surely winning strategy to access the bank account. For example if \Low's observation after executing $chkWeb$ is $\obsL{GNameC}$, the actual state might be $s_{11}$ or $s_{13}$, which need different actions for \Low to access the account. However in variant (b), after executing $chkWeb$, \Low's observation will allow her to have a sure winning strategy to access the account. For example if \Low's observation after execution of $chkWeb$ is $\obsL{MnameA}$, the actual state is either $s_{11}$ or $s_{12}$, and in both \Low can access the account by executing $auth_A$.}%end-extended
%}
%\label{fig:banking2}
%\end{figure}

\begin{example}\label{ex:ni}
Consider a simplified version of the phone banking scenario from \extended{Section~\ref{sec:motivating}}\short{Example~\ref{ex:motivating}}. There are two users: \High who has an account in the bank, and \Low who may try to impersonate \High. \High can access her account by correctly giving the maiden name of her mother. Moreover, \High runs a blog, and can publish some of her personal information on it. We consider two alternative variants: one where \High publishes her grandmother's maiden name on the blog (Figure~\ref{fig:banking2a}), and one where she publishes her mother's maiden name (Figure~\ref{fig:banking2b}). We assume that the possible names are $A$ and $B$ in the former case, and $C$ and $D$ in the latter.\short{Each model begins by initialization of the relevant names.}
 %Figure~\ref{fig:banking2} shows the transition networks for the two variants.
%Labels on transitions show the personalized actions resulting in the transition, and the observations of users in each state are shown beside the state.
The observations of \Low are shown beside each state. The observations for \High are omitted, as they will be irrelevant for our analysis.
%Also only the relevant  transitions are depicted, all the transitions that are not shown are consider to be reflexive.
\extended{
  % **********

  Each model begins by initialization of the relevant names, represented by virtual actions of agent $\Environ$. Then, \High publishes an essay on her blog. In the first variant, the essay mentions the maiden name of \High's grandmother. In the second variant, it mentions her mother's maiden name.
  After H has published the essay, \Low can check the blog (action $\chkWeb$). The resulting observation of \Low depends on what is published.
  Then, authentication proceeds\extended{ like in Example~\ref{ex:banking1}}: in order to log in, a user must give the correct value of \High's mother's maiden name.
  % **********
}%end extended

Note that, for mathematical completeness, we must define the outcome of every user-action pair in every state. We assume that there are two ``error states'' $\sherr,\slerr$ in models $M_a$ and $M_b$ (not shown in the graphs). Any action of \High not depicted in the figure leads to $\sherr$, and any action of \Low not depicted in the figure leads to $\slerr$.
We will later use the error states in the definition of the players' goals, in such a way that \Low will always want to avoid $\slerr$ and \High will want to avoid $\sherr$.
This way we can (however imperfectly) simulate some synchronization in the restricted framework of Goguen and Meseguer.

Neither $M_a$ nor $M_b$ satisfies noninterference from \High to \Low. For instance, in the model of Figure~\ref{fig:banking2a}, if $\alpha = \langle (H,\Mname_A),(H,\Gname_D),(H,publish),(L,\chkWeb)\rangle$, the observation of \Low after sequence $\alpha$ is $\obsL{GNameD}$, but the observation of \Low after $\Purge_\High(\alpha)=\langle (L,\chkWeb) \rangle$ is $\obsL{noObs}$, which is clearly different.
\qed
\end{example}

\extended{
  % **********
  Again, two remarks are in order.
  First, noninterference focuses solely on the information flow in the system. If \Low can detect any activity of \High then noninterference is lost, regardless of the nature of the activity and the possible uses of the information. In real systems, the impact of information flow goes well beyond the information itself. Information is sought and preserved for a reason, not for its own sake. Typically, \Low want to obtain information about \High because they want to use it to achieve their goals more effectively (i.e., conclude a business contract, submit a better bid in an auction, get unauthorised access to a bank account etc.). On the other hand, \High want to protect their private information from \Low because their goals may be in conflict with the goals of \Low. This is especially the case when the Low players are labeled as ``attackers'' or ``intruders''.

  Secondly, detecting \High's actions may require \Low to engage in ``diagnostic'' activity, i.e., executing a sequence of actions whose only purpose is to determine if \High was active or not. This becomes an issue when we see information as a resource used to obtain one's goals, rather than \emph{the} goal of the user's activity. Then, obtaining more information about \High can be in conflict with what \Low must do in order to achieve their real goals. Thus, on one hand \Low need more information to construct a better strategy for their goals, but on the other hand to acquire the information they may have to depart from the successful strategy.
  % **********
}%end extended

%In this paper, we aim at generalizing the concept of noninterference so that it takes into account the strategic aspect of information, seen as a resource relevant for obtaining other, not necessarily information-related goals.

\WJ{note the connection to distributed knowledge within $\Low$}

%Assuming that $\High$ need to hide only occurrences of some ``sensitive'' actions $A\subseteq\Act$, the concept of noninterference is refined as follows.

%\begin{definition}[Noninterference on sensitive actions~\cite{goguen1982security}]
%Given a transition network $M$, sets of agents $\High,\Low$, and a set of actions $A\subseteq\Act$, we say that \emph{$\High$ is non-interfering with $\Low$ on $A$} iff for all $\alpha\in(\Agt\times\Act)^*$ and all $\low\in \Low$ we have $[exec(\alpha)]_{\low} = [exec(\Purge_{\High,A}(\alpha))]_{\low}$. We denote the property by $\NI_M(\High,\Low,A)$.
%\end{definition}
%
%\para{Related Concepts of Secure Information Flow.}
%Since the introduction of noninterference, a number of variants and extensions have been proposed and studied~\cite{goguen1982security,sutherland1986model,o1990calculus,mccullough1988noninterference,Wittbold90information,Jamroga15strat-ni,%
%allen1991comparison,roscoe1994non,roscoe1995csp,roscoe1999intransitive,ryan2001process,mccullough1988noninterference,zakinthinos1995composability,%
%seehusen2009information,backes2003intransitive,van2007indeed,engelhardt2012intransitive,gray1990probabilistic,mciver2003probabilistic,%
%pierro2004approximate,li2005downgrading,smith2009foundations}.
%We discuss the formalizations in some more detail in Section~\ref{sec:relatedwork}.
%
%It should be mentioned that, while we use the original noninterference as the starting point here, the reasoning scheme that we propose can be similarly applied to other concepts of secure information flow.

\subsection{Strategies and Their Outcomes}

\emph{Strategy} is a game-theoretic concept which captures behavioral policies that an agent can consciously follow in order to realize some objective~\cite{\extended{vonNeumann44,}LeytonBrown08GT}.
\extended{
  % **********
  We begin with an abstract formulation, and mention the most representative examples of strategy types in the next paragraph.
  % **********
}%end extended
Let $T(M)$ be the \emph{tree unfolding} of $M$. Also if $U\subseteq\Agt$ is a subset of agents, let $T'$ be a \emph{$U$-trimming} of tree $T$ iff $T'$ is a subtree of $T$ starting from the same root and obtained by removing an arbitrary subset of transitions labeled by actions of agents from $U$. For the moment, we assume that each subset of agents $U\subseteq\Agt$ is assigned a set of available coalitional strategies $\Sigma_U$. The most important feature of a strategy $\sigma_U \in \Sigma_U$ is that \emph{it constrains the possible behaviors of the system}. We represent it formally by the \emph{outcome function} $out_M(\sigma_U)$ that removes the executions of the system that strategy $\sigma_U$ would never choose. Therefore, for every $\sigma_U\in\Sigma_U$, its outcome $out_M(\sigma_U)$ is a $U$-trimming of $T(M)$.
\WJ{Do we define any constraints on $out_M$? In particular, it should not constrain the actions of the other players...}

Let $h$ be a node in tree $T$ corresponding to a particular finite history of interaction. We denote the sequence of personalized actions leading to $h$ by $act^*(h)$. Furthermore, $act^*(T) = \set{ act^*(h) \mid h\in nodes(T) }$ is the set of finite sequences of personalized actions that can occur in $T$.

%Let $h$ be a node in tree $T(M)$. We denote the history of interactions corresponding to $h$ by $hist(h)\in s_0.(\alpha_i.s_i)^*$, where $s_i\in \States$ and $\alpha_i \in (\Agt\times\Act)$. This history is a sequence of states $s_i$, starting from $s_0$, and interleaved by personalized actions $\alpha_i$, which shows the path through the tree $T(M)$ starting from the root and finishing at $h$. We denote the sequence of personalized actions leading to $h$ by $act^*(h)$. Furthermore, $act^*(T) = \set{ act^*(h) \mid h\in nodes(T) }$ is the set of finite sequences of personalized actions that can occur in $T$.

\extended{\para{Types of strategies.}}
Strategies are usually constructed as mappings from possible situations that the player can recognize in the game, to actions of the player (or subsets of actions if we allow for nondeterministic strategies).
\extended{%begin extended
  % **********
  Two types of such strategies are commonly used in the literature on game-like interaction: positional strategies and perfect recall strategies. Positional strategies represent conditional plans where the decision is solely based on what the agents see in the current state of the system, while perfect recall strategies capture conditional plans where the agents can base their decisions on the whole history of the game until that moment. 

  \emph{Positional strategies} represent conditional plans where the decision is solely based on what the agents see in the current state of the system. Formally, for $u\in\Agt$, the set of individual positional strategies of $u$ is $\Sigma_u^\norecl = \{ \sigma_u : \States\then\powerset{\Act}\setminus\set{\emptyset}\ \mid\ \forall q,q'\in \States \	\cdot\	[q]_u=[q']_u \Rightarrow \sigma_u(q)=\sigma_u(q') \}$, where $\powerset{X}$ denotes the powerset of $X$. Notice the ``uniformity'' constraint which enforces that the agent must specify the same action(s) in states with the same observations.
  Now, coalitional positional strategies for a group of agents $U\subseteq\Agt$ are simply tuples of individual strategies, i.e., $\Sigma_U^\norecl = \times_{u\in U}(\Sigma_u^\norecl$).
  The outcome of $\sigma_U\in\Sigma_U^\norecl$ in model $M$ is the tree obtained from $T(M)$ by removing all the branches that begin from a node containing state $q$ with a personalized action $(u,a)\in U\times\Act$ such that $a\notin\sigma_U(q)$.

  In this work we focus on adversaries playing perfect recall strategies.
  
  \emph{Perfect recall strategies.}
  % **********
}%end extended
Formally, the set of \emph{perfect recall strategies} of agent $u$ is $\Sigma_u^\recl = \{ \sigma_u : nodes(T(M))\then\powerset{\Act}\setminus\set{\emptyset}\ \mid\ obs_u(h) = obs_u(h') \Rightarrow \sigma_u(h)=\sigma_u(h') \}$, where $obs_u(h)$ denotes the accumulate observations collected by agent $u$ along history $h$.
How to define $obs_u$ for sequences of states?
For asynchronous systems, this is typically defined as $obs_u(q) = [q]_u$, $obs_u(h\circ q) = obs_u(h)$ if $last(h)=q$, and $obs_u(h\circ q) = obs_u(h)\circ[q]_u$ otherwise (where $\circ$ denotes the concatenation operator).
That is, what $u$ has learned along $h$ is equivalent to the sequence of observations she has seen, modulo removal of ``stuttering'' observations.
%Note again the ``uniformity'' constraint which enforces that the agent must specify the same action(s) for indistinguishable situations (i.e., sequences of states from $M$).
%That is, the ``uniformity'' constraint is now based on the assumption that the agent knows the history of its own observations, plus the actions that it has executed.
Now, coalitional strategies of perfect recall for a group of agents $U\subseteq\Agt$ are combinations of individual strategies, i.e., $\Sigma_U^\recl = \times_{u\in U}(\Sigma_u^\recl$).
The outcome of $\sigma_U\in\Sigma_U^\recl$ in model $M$ is the tree obtained from $T(M)$ by removing all the branches that begin from a node $h$ with a personalized action $(u,a)\in U\times\Act$ such that $a\notin\sigma_U(h)$.

%\todo{EXAMPLE + FIGURE (tree unfolding of web banking model with highlighted paths enabled by a given strategy}

\subsection{Temporal Goals and Winning Strategies}

A goal is a property that some agents may attempt to enforce by selecting their behavior accordingly.
\extended{
  % **********
  In game-theoretic models, goals are typically phrased as properties of the final state in the game.
  In our case, there is no final state -- the interaction can go on forever.
  Because of that, we understand goals as properties of the full temporal trace that executes the sequence of actions selected by users.
  % **********
}%end extended
We base our approach on the concepts of \emph{paths} and \emph{path properties}, used in temporal specification and verification of systems~\cite{\extended{Buchi62omega-regular,}McNaughton66omega-regular}.
%
%Let $traces(M)$ be the set of finite or infinite sequences of states that can be obtained by subsequent transitions in $M$.
%Moreover, $paths(M)$ will denote the set of maximal traces, i.e., those sequences that are either infinite or end in a state with no outgoing transitions. Additionally, we will use $paths_M(\sigma)$ as a shorthand for $paths(out_M(\sigma))$.
%Note that, in the models of Goguen and Meseguer, strategies of any group except for the grand coalition $\Agt$ yield only infinite paths. Since we will only look at the goals of \emph{subsets} of players, it suffices to restrict our attention to infinite paths.
Let $paths(M)$ denote the set of infinite sequences of states that can be obtained by subsequent transitions in $M$.
Additionally, we will use $paths_M(\sigma)$ as a shorthand for $paths(out_M(\sigma))$.
%
%\todo{check if the definition of the goal can remain the same with the new model?
%WJ: I DON'T SEE ANY PROBLEMS.}
\begin{definition}[Temporal goal~\cite{McNaughton66omega-regular}]\label{def:goal}
A \emph{goal} in $M$ is any $\Gamma\subseteq paths(M)$.
Note that $paths(M) = paths(T(M))$, so a goal can be equivalently seen as a subset of paths in the tree unfolding of $M$.
%\WJ{We need finite traces as well, since we allow for strategies that specify empty set for some situations -- thus it might be the case that some path is ``blocked'' by the strategy though it was infinite in the model.}
\end{definition}

%We will typically assume goals to be sets of paths definable in Linear Temporal Logic~\cite{Schnoebelen03complexity}.
Most common examples of such goals are safety and reachability goals.
%For example, a goal of user $u_1$ can be that message $m$ is, at some future moment, communicated to user $u_2$. Or, the users $u_1$ and $u_2$ may have a joint goal of keeping the communication channel $c$ operative all the time. The former is an example of a \emph{reachability goal}, the latter a \emph{safety goal}.

\begin{definition}[Safety and reachability goals~\cite{McNaughton66omega-regular}]\label{def:safety-reachability}
Given a set of safe states $\safest\subseteq \States$, the \emph{safety goal $\Gamma_\safest$} is defined as $\Gamma_\safest=\set{\lambda\in paths(M) \mid \forall i . \lambda[i]\in\safest}$.
%, where $\lambda(i)$ is the prefix sequence of $i$ element of $\lambda$.
Moreover, given a set of target states $\reachst\subseteq \States$, the \emph{reachability goal $\Gamma_\reachst$} can be defined as $\Gamma_\reachst=\set{\lambda\in paths(M) \mid \exists i . \lambda[i]\in\reachst}$.
\end{definition}

\begin{definition}[Winning strategies]
Given a transition network $M$, a set of agents $U\subseteq\Agt$ with goal $\Gamma_U$, and a set of strategies $\Sigma_U^\recl$, we say that $U$ have a (surely winning) strategy to achieve $\Gamma_U$ iff there exists a strategy $\sigma_U\in\Sigma_U^\recl$ such that $paths_M(\sigma_U) \subseteq \Gamma_U$.
%Given a transition network $M$, a set of High agents $\High$, a set of Low agents $\Low$ with goal $\Gamma_\Low$, and a set of ``sensitive'' actions $A$, we say that \Low have a sure winning strategy to achieve $\Gamma_\Low$ iff there exists a (perfect recall) strategy $\sigma_\Low$ for $\Low$, such that $paths_M(\sigma_\Low) \subseteq \Gamma_\Low$. We can write this as $SW_L(M,A,\Gamma_\Low)$.
\end{definition}

%\todo{EXAMPLE: reachability goal for Low (eventually Low breaks into High's account or the system goes to $s_{Herr}$). Demonstrate that there is a winning strategy for Low in one model, but no winning strategy in the other model.}

\begin{example}\label{exmpl:goal}
Consider the models in Figure~\ref{fig:banking2a} and Figure~\ref{fig:banking2b}, and suppose that \Low wants to access \High's bank account. This can be expressed by the reachability goal $\Gamma_\reachst$ with $\reachst = \set{s_{15},s_{16}}$ as the target states.
In fact, \Low also wins if \High executes an out-of-place action (cf.~Example~\ref{ex:ni} for detailed explanation).
In consequence, the winning states for \Low are $\reachst = \set{s_{15},s_{16},\sherr}$.
Note that \Low has no strategy that guarantees $\Gamma_\reachst$ in model $M_a$ (although information is theoretically leaking to \Low as the model does not satisfy noninterference). Even performing the action $\chkWeb$ does not help, because \Low cannot distinguish between states $s_{11}$ and $s_{13}$, and there is no single action that succeeds for both $s_{11},s_{13}$. Thus, \Low does not know whether to use $auth_A$ or $auth_B$ to get access to \High's bank account.

On the other hand, \Low has a winning strategy for $\Gamma_\reachst$ in model $M_b$. The strategy is to execute $\chkWeb$ after \High publishes her mother's maiden name, and afterwards do $auth_A$ in states $s_{11},s_{12}$ (after observing $\obsL{MNameA}$) or $auth_B$ if the system gets to $s_{13},s_{14}$ (i.e., after observing $\obsL{MNameB}$).
\qed
\end{example}

\extended{
  % **********
  In what follows, we will look at the \Low's strategic ability to harm desirable behavior of the system.
  % **********
}%end extended

\section{Security as Strategic Property}\label{sec:effsec}

The property of noninterference looks for \emph{any} leakage of \emph{any} information. If one can possibly happen in the system, then the system is deemed insecure. In many cases, this view is too strong. There are lots of information pieces that can leak out without bothering any interested party. Revealing the password to your web banking account can clearly have much more disastrous effects than revealing the price that you paid for metro tickets on your latest trip to Paris. Moreover, the relevance of an information leak cannot in general be determined by the type of the information.
Think, again, of revealing the maiden name of your mother vs.~the maiden name of your grandmother. The former case is potentially dangerous since the maiden name of one's mother is often used to grant access to manage banking services by telephone. Revealing the latter is quite harmless to most ends and purposes.\extended{ \footnote{
  Note, however, that revealing the maiden name of your maternal grandmother is potentially dangerous to your \emph{mother} if she enables banking by telephone. }
}%end extended

In this paper, we suggest that the relevance of information leakage should be judged by the extent of damage that the leak allows the attackers to inflict on the goal of the system.
Thus, as the first step, we define the security of the system in terms of damaging abilities of the Low players.

%\subsection{Security as Strategic Property}

In order to assess the relevance of information flow from High to Low, we will look at the resulting strategic abilities of Low.
\extended{
  % **********
  For this, two design choices have to be made. First, what type of strategies are Low supposed to use? Secondly, what is the goal that they are assumed to pursue? The second question is especially important, because typically we do not know (and often do not care about) the real goals of potential attackers. What we know, and what we want to protect, is the objective that the system is built for.

  We follow the game-theoretic tradition of looking at the worst case and assuming the opponents to be powerful and adversary.
  Thus, we assume \Low to use perfect recall strategies.
  % **********
}%end extended
Moreover, we assume that the goal of \Low is to violate a given goal of the system.
The goal can be a functionality or a security requirement, or a combination of both. Moreover, it can originate from a private goal of the High players, an objective ascribed to the system by its designer (e.g., the designer of a contract signing protocol), or a combination of requirements specified by the owner/main stakeholder in the system (for instance, a bank in case of a web banking infrastructure).

\begin{definition}[Effective security]
Let $M$ be a transition network with some Low players $\Low\subseteq\Agt$, and let $\Gamma$ be the goal of the system.
We say that $M$ is \emph{effectively secure} for $(\Low,\Gamma)$ iff $\Low$ does not have a strategy to enforce $\overline{\Gamma}$, where $\overline{X}$ denotes the complement of set $X$.
That is, the system is effectively secure iff the attackers do not have a strategy that ensures an execution violating the goal of the system.
We will use $ES(M,\Low,\Gamma)$ to refer to this property.
%not $\Win_\Low(M,\overline{\Gamma})$.
\end{definition}

Besides judging the effective security of a system, we can also use the concept to compare the security level of two models.

\begin{definition}[Comparative effective security]\label{def:compeffsec}
Let $M,M'$ be two models, and $\Gamma$ be a goal in $M,M'$ (i.e., $\Gamma\subseteq paths(M)\cup paths(M')$).
We say that:
\begin{itemize2}
\item $M$ has \emph{strictly less effective security} than $M'$ for $(\Low,\Gamma)$, denoted $M \prec_{\Low,\Gamma} M'$, iff $ES(M',\Low,\Gamma)$ but not $ES(M,\Low,\Gamma)$.
  \extended{
    % **********
    That is, \Low can enforce a behavior of the system that violates its goal in model $M$ but not in $M'$.
    We denote the relationship by $M \prec_{\Low,\Gamma} M'$;
    % **********
  }%end extended

  \item $M'$ is \emph{at least as effectively secure as $M$} for $(\Low,\Gamma)$, denoted $M \preceq_{\Low,\Gamma} M'$, iff $ES(M,\Low,\Gamma)$ implies $ES(M',\Low,\Gamma)$;

\item $M$ is \emph{effectively equivalent to $M'$} for $(\Low,\Gamma)$, denoted $M \simeq_{\Low,\Gamma} M'$, iff either both $ES(M,\Low,\Gamma)$ and $ES(M',\Low,\Gamma)$ hold, or both do not hold.
\end{itemize2}
\end{definition}

Thus, if in one of the models \Low can construct a more harmful strategy then the model displays lower {effective security} than the other model. Conversely, if both models allow only for the same extent of damage then they have the same level of effective security. This way, we can order different alternative designs of the system according to the strategic power they give away to the attacker.

\begin{example}\label{exmpl:effective-secure}
Consider models $M_a, M_b$ from Figure~\ref{fig:banking2a} and Figure~\ref{fig:banking2b}, and let the goal $\Gamma$ be to prevent \Low from accessing \High's bank account. Thus, $\Gamma$ is the safety goal $\Gamma_\safest$ with $\safest = \States\setminus\set{s_{15},s_{16},\sherr}$, and therefore $\overline{\Gamma} = \Gamma_\reachst$  with $\reachst = \set{s_{15},s_{16},\sherr}$.
As we saw in Example~\ref{exmpl:goal}, \Low has no strategy to guarantee $\overline{\Gamma}$ in $M_a$, but she has a surely winning strategy for $\overline{\Gamma}$ in $M_b$. Thus, $M_b$ is strictly less effectively secure than $M_a$, i.e., $M_b \prec_{\Low,\Gamma} M_a$.
\end{example}

\extended{
  % **********
  We will further use the concept to compare security of alternative information flows based on the same (or similar) action-transition structures.
  % **********
}%end extended

\putaway{
  %**********
  \subsection{Comparing Effective Security of Information Flows}

  In this paper, we focus on the following question: How to design the information flow in a system so that no strategically important information is revealed? In the simplest case, we may want to compare different information flows that are possible in a given action-transition structure. To this end, we define the notion of transition equivalence.

  \begin{definition}[Transition-equivalent models]
  The \emph{action-transition frame} of a model $M$, which we denote by $F_M$, is the network $M$ minus the observation functions $obs(\cdot)$.
  We will denote the set of models based on frame $F$ by $\modl(F)$.
  Two models are \emph{transition-equivalent} iff they are based on the same frame.
  \end{definition}

  Now, two models differing only in the informational aspect can be compared by the damaging abilities of \Low, according to Definition~\ref{def:compeffsec}.

  \begin{example}
  The two models $M_a$ and $M_b$ of Figure~\ref{fig:banking2a} and Figure~\ref{fig:banking2a} are transition-equivalent. All the possible actions and transitions in the two cases of the scenario are the same, however the content of information revealed to \Low which affects her observations in some states (and eventually her strategic abilities) are different in the two scenarios.
  \end{example}

  Clearly, if \Low's observation function in $M$ refines the one in $M'$ then $M'$ offers at least as much effective security as $M$ for any goal.

  \begin{proposition}\label{prop:refines}
  Let $M,M'$ be transition-equivalent models such that, for any pair of states $s_1,s_2\in\States$, if $[s_1]_\Low^M = [s_2]_\Low^M$ then $[s_1]_\Low^{M'} = [s_2]_\Low^{M'}$.
  Then, for every goal $\Gamma \subseteq paths(M)$, we have that $M \preceq_{\Low,\Gamma} M'$.
  \end{proposition}
  \begin{proof}
  Note that all the strategies of \Low in $M'$ are also \Low's strategies in $M$. Thus, if \Low have a surely winning strategy to enforce $\overline{\Gamma})$ in $M'$ then they also have a surely winning strategy for $\overline{\Gamma})$ in $M$.
  \end{proof}
  %**********
}% end-putaway

\section{Effective Information Security}\label{sec:effinfosec}

\newcommand{\Ideal}{\mathit{Ideal}}

%In this paper, we focus on the following question: How to design the information flow in a system so that no strategically important information is revealed?
We will now propose a scheme that allows to determine whether a given model of interaction leaks relevant information or not.
We use the idea of refinement checking from process algebras, where a process is assumed correct if and only if it refines the ideal process~\cite{Roscoe97csp}.
A similar reasoning scheme is also used in analysis of multi-party computation protocols (a protocol is correct iff it is equivalent to the ideal model of the computation~\cite{Goldreich87mentalgame}).

\extended{
  % **********
  To this end, we need a suitable notion of refinement or equivalence, and a suitable definition of the ideal model.
  The former is straightforward: we will use the $\simeq_{\Low,\Gamma}$ relation.
  The latter is more involved. If the reference model ascribes too much observational capabilities to the Low players then the concept will be ill-defined (it will classify insecure systems as secure). If the reference model assigns Low with too little information then the concept will be useless (no realistic system will be ever classified as secure).
  %Thus, the ideal model must provide the right abstraction level

  In what follows, we first explain and define the concept of an idealized variant of a model. Then in Section~\ref{sec:blinding-naive} we do our first take on the idealized variant by defining the \textit{blind variant} of a model. In Section~\ref{sec:blind-noninterf} we first define the \textit{non-interfering idealized} variant of a model.
  % **********
}%end extended

\subsection{Ability-Based Security of Information Flows}

Definition~\ref{def:compeffsec} allows for comparing the effective security of two alternative information flows.
\extended{
  % **********
  We will say two models differ only in their information flow if they are \emph{transition equivalent}:

  \begin{definition}[Transition-equivalent models]
    The \emph{action-transition frame} of a model $M$, which we denote by $F_M$, is the network $M$ minus the observation functions $obs(\cdot)$.
    We will denote the set of models based on frame $F$ by $\modl(F)$.
    Two models are \emph{transition-equivalent} iff they are based on the same frame.
  \end{definition}

  Then , $(F,obs)\prec (F,obs')$ says that the observation function $obs$ leaks more relevant information than $obs'$ in the transition-action frame $F$.
  % **********
}%end extended
However, we usually do not want to compare several alternative information flows. Rather, we want to determine if a single given model $M$ reveals relevant information or not.
\extended{How can we achieve that? }A natural idea is to compare the effective security of $M$ to an \emph{ideal model}, i.e.\short{ a variant of $M$ that }\extended{a model that is transition equivalent to the original model and moreover }leaks no relevant information by construction.
Then, a model is effectively information-secure if it has the same level of effective security as its idealized variant:

\begin{definition}[Effective information security]
Let $M$ be a transition network with some Low players $\Low\subseteq\Agt$, and let $\Gamma$ be the goal of the system.
Moreover, let $\Ideal(M)$ be the idealized variant of $M$.
We say that $M$ is \emph{effectively information-secure} for $(\Low,\Gamma)$ iff $M \simeq_{\Low,\Gamma} \Ideal(M)$.
\end{definition}

%Having a transitional network $M$, how we can know if we are able to increase the security of the system by reducing the strategic ability of the players \Low by modifying its epistemic abilities? One way to approach this question is comparing the strategic ability of players in a given model, with their strategic ability in an idealized model in which the \Low players have the minimum epistemic ability. If these two models are similarly s-secure (or $E_n$ secure), we may argue that reducing the amount of information leaked to \Low players wouldn't change their strategic ability to achieve their goal, and therefore is not necessary or helpful. On the other hand if the model is less s-secure than the idealized model, it shows that it is possible to reduce the strategic ability of the \Low players by reducing the amount of leaking information to them.

How do we construct the idealized variant of $M$? The idea is to ``blur'' observations of Low so that we obtain a variant of the system where the observational capabilities of the attackers are minimal.
\extended{
  % **********
  What observational capabilities are ``minimal''?
  We start with the following, rather naive definition of idealization.

  \subsection{Blinding the Low Players: First Attempt}\label{sec:blinding-naive}

  By using the idealized model, we intend to distinguish to what extent the damaging abilities of Low are due to the ``hard'' actions available in the system, and to what extent they are due to the available information flow.
  In other words, we want to see how far one can minimize the strategic ability of the Low players by reducing their observational abilities in the model.

  %For reducing the observational abilities of the Low player, the first step is to take away the information it can get from observing it's own available actions in a state. Because however we try to obscure the observations of a player in a state, it still may gain information just by ``observing'' its available actions.

  The first take to define an idealized model is to assume that \Low \emph{never sees anything}. To this end, we simply assume that $obs(s,\Low)$ is the same in all states $s\in\States$.

  \begin{definition}[Idealized model, first take]
  Having a transition network $M$ based on frame $F$, and a set of players \Low, we define the \emph{blind variant of $M$} as $M' = (F,obs')$ such that $obs'(q,l)=obs(q',l)$ for every $q,q'\in\States$ and $l\in\Low$.
  \end{definition}

  %\begin{definition}[Idealized model, first take]
  %Let $F'=total_L(M)$ be the action transition framework of the L-total variant of $M$. We define the \emph{blind variant of $M$} as $M' = (F',obs')$ such that $obs'(q,l)=obs'(q',l)$ for every $q,q'\in\States$ and $l\in\Low$.
  %\end{definition}

  In many scenarios this is too much. In particular, a Low agent may have access to perfectly legitimate observations that are inherent to maintaining their private affairs, such as checking the balance of their bank account, listing the files stored on in their private file space, etc.

  %We say that a model $M$, is as secure as the blinded idealized model for the players $\Low$ with the goal $\Gamma_\Low$, iff $M \approx_{\Low,\Gamma_\Low} M'$, for some $M'\in \Ideal^{\blind}_L(M)$. We say that it is less secure than the blinded idealized model iff $M \prec_{\Low,\Gamma_\Low} M'$.

  %In a blinded model the \Low players cannot distinguish between any two states of the system using their observations in those states. Let's say the \Low players have a winning strategy even in this completely obscured model. Then clearly they would have a winning strategy in the original model too. Moreover we can say that no matter how some empowered designer of the system try to obscure the events and the observations from the Low player, the system will remain insecure. In other words in this case the system is inherently insecure.

  %By using the idealized model, we intend to distinguish to what extent the damaging abilities of Low are due to the ``hard'' actions available in the system, and to what extent they are due to the available information flow.
  %In other words, we want to see how far one can minimize the strategic ability of the Low players by reducing their observational abilities in the model. The blinded variant does it by assuming that \Low \emph{never see anything}.
  % **********
}%end extended

\subsection{Idealized Models Based on Noninterference}\label{sec:blind-noninterf}

\extended{
  % **********
  Below we propose a weaker form of ``blinding'' that will be used to single out the damaging abilities that are due to Low \emph{observing High's actions}, rather than due to \emph{any} observations that Low can happen to make.
  % **********
}%end extended
We begin by recalling the notion of \emph{term unification} which is a fundamental concept in automated theorem proving and logic programming~\cite{Robinson65resolution}.
Given two terms $t_1,t_2$, their unification ($t_1\equiv t_2$) can be understood as a declaration that, from now on, both terms refer to exactly the same underlying object. In our case the terms are observation labels from the set $Obs$. A unification can be seen as an equivalence relation on observation labels, or equivalently as a partitioning of the labels into equivalence classes. The application of the unification to a model yields a similar model where the equivalent observations are ``blurred''.

\begin{definition}[Unification of observations]
Given a set of observation labels $Obs$, a \emph{unification on $Obs$} is any equivalence relation $\mathcal{U} \subseteq Obs\times Obs$.

Given a model $M = \tuple{\States,s_0,\Agt,\Act,do,Obs,obs}$ and a unification $\mathcal{U} \subseteq Obs\times Obs$, the \emph{application of $\mathcal{U}$ to $M$} is the model $\mathcal{U}(M) = \tuple{\States,s_0,\Agt,\Act,do,Obs',obs'}$, where:
\ $Obs' = \set{[o]_{\mathcal{U}} \mid o\in Obs}$ replaces $Obs$ by the set of equivalence classes defined by $\mathcal{U}$, and
\ $obs'(q,u) = [obs(q,u)]_{\mathcal{U}}$ replaces the original observation in $q$ with its equivalence class for any $u\in\Agt$.
\end{definition}

%\todo{EXAMPLE (unification in the web banking model)}
\extended{%begin extended
  % **********
  \begin{figure*}[!t]\centering
  \hspace*{-30pt}
  \begin{tikzpicture}[thick,scale=0.6, every node/.style={transform shape}]
  \input{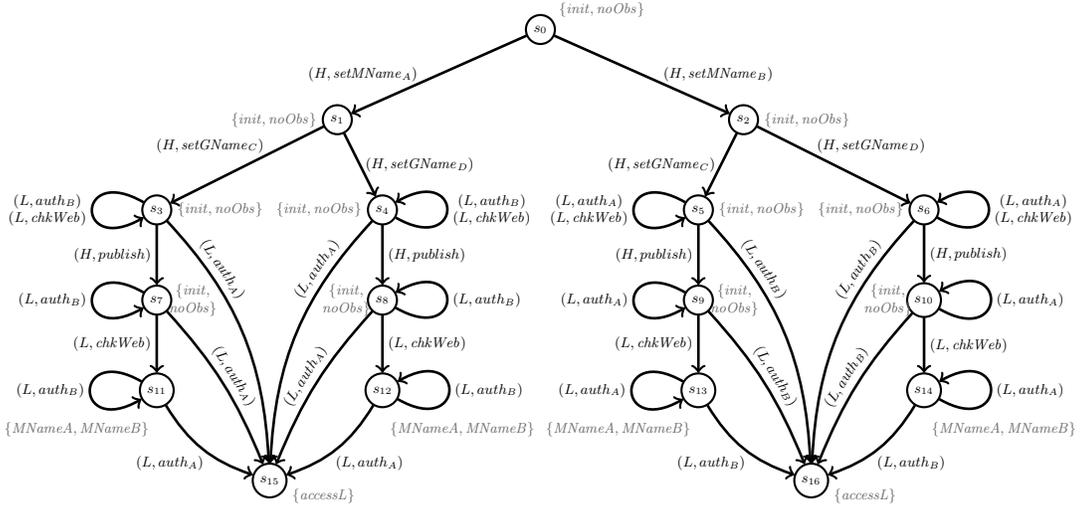}
  \end{tikzpicture}
  \caption{An example of unification of observations.}
  \label{fig:banking-unification}
  \end{figure*}

  \begin{example}
  Figure~\ref{fig:banking-unification} depicts the model obtained from $M_b$ by unifying observations $\obsL{MNameA}$ and $\obsL{MNameB}$ into $\set{\obsL{MNameA},\obsL{MNameB}}$, observations $\obsL{init}$ and $\obsL{noObs}$ into $\set{\obsL{init,noObs}}$, and observation $\obsL{accessL}$ into $\set{\obsL{accessL}}$.
  \end{example}
  % **********
}%end extended

Our reference model for $M$ will be the variant of $M$ where noninterference is obtained by the minimal necessary ``blurring'' of \Low's observations.

\begin{definition}[Noninterferent idealized model]\label{def:NI-idealized}
Having a transition network $M$ and a set of ``low'' players \Low, we define the \emph{noninterfering idealized variant of $M$} as $\mathcal{U}(M)$ such that:
\begin{enumerate}
\item[(i)] $\NI_{\mathcal{U}(M)}(\High,\Low)$, and
\item[(ii)] for every $\mathcal{U}' \subsetneq \mathcal{U}$ it is not the case that $\NI_{\mathcal{U}'(M)}(\High,\Low)$.
\end{enumerate}
%
%We will denote $\Ideal^{\NIblind}_L(M)=\set{M'\in tEQ(M)\mid \NI_{M'}(\Agt\setminus\Low,\Low,S)}$.
\end{definition}

We need to show that the concept of noninterferent idealized model is well defined. The proof is constructive, i.e., given a model $M$, we first show how one can build its idealized variant, and then show that it is unique.%
\short{\footnote{
  We will only sketch the proofs due to lack of space. The complete proofs can be found in the extended version of the paper, available on  \url{https://www.arxiv.org} . }
}% end-short

\begin{theorem}\label{prop:uniqueness}
For every transition network $M$, there is always a unique unification $\mathcal{U}$ satisfying properties (i) and (ii) from Definition~\ref{def:NI-idealized}.
\end{theorem}
\extended{
%**********
The proof of Theorem~\ref{prop:uniqueness} needs some preliminary steps. As the first step, we recall and adapt the concept of unwinding relations~\cite{\extended{goguen1984unwinding,rushby1992noninterference,}Meyden10comparison}. Unwinding is constructed analogously to the standard notion of bisimulation, and requires Low's uncertainty to be a fixpoint of an appropriate relation transformer.
Unwinding relations are important because they characterize \NItxt in purely structural terms. Moreover, existence of an unwinding relation is usually easier to verify than proving \NItxt directly. We then use the concept of unwinding relation to define relation $\rstar$ on the states of a transition network $M$. We use this relation to construct and prove the uniqueness of the idealized variant of $M$.
}%end extended
\short{
The proof of Theorem~\ref{prop:uniqueness} needs some preliminary steps. First, we recall the concept of unwinding relations~\cite{\extended{goguen1984unwinding,rushby1992noninterference,}Meyden10comparison}.
Unwinding relations are important because they characterize \NItxt in purely structural terms. Moreover, existence of an unwinding relation is usually easier to verify than proving \NItxt directly. We then use the concept of unwinding relation to define relation $\rstar$ on the states of a transition network $M$. We use this relation to construct and prove the uniqueness of the idealized variant of $M$.
}%end short
\begin{definition}[Unwinding for Noninterference~\cite{Meyden10comparison}]\label{def:unwind-nonint}
\extended{Let $M$ be a transition network, $\High$ a set of High agents, and $\Low$ a set of Low agents. Then,}
$\uwNIL \subseteq \States \times \States$ is an \emph{unwinding relation} iff it is an equivalence relation satisfying the conditions of \emph{output consistency (OC)}, \emph{step consistency (SC)}, and \emph{local respect (LR)}. That is, for all states $s,t \in \States$:
\begin{description}
\item[(OC)] If $s\uwNIL t$ then $[s]_\Low=[t]_\Low$;

\item[(SC)] If $s\uwNIL t$, $u\in \Low$, and $a\in \Act$ then $do(s,u,a)\uwNIL do(t,u,a)$;

\item[(LR)] If $u\in \High \textrm{ and } a\in \Act$ then $s \uwNIL do(s,u,a)$.
\end{description}
\end{definition}
%The above definition of unwinding, and the following correspondence theorem, are similar to the ones in~\cite{goguen1984unwinding,rushby1992noninterference}. However, we had to extend them to handle a broader class of models which are not necessarily total on input.
\begin{proposition}[\cite{Meyden10comparison}]\label{prop:unwind-nonint}
%In a transition network $M$, with set of High agents $\High$ and a set of Low agents $\Low$, we have
$NI_M(\High,\Low)$ iff there exist an unwinding relation $\uwNIL$ on the states of $M$ that satisfies (OC), (SC) and (LR).
\end{proposition}
%\begin{proof}
%See the Appendix.
%\end{proof}

Next we define $\rstar$ on the states of a transition network $M$. The definition goes as follows: first we relate any two states of $M'$  if one of them can be reached from the other one by a sequence of High personalized actions. Then in each step we relate the pair of states that are reached by a similar Low personalized action from any two states that are already related. Also, we enforce transitivity on the set. We continue adding related states until the relation becomes stable.
% **********
The mathematical definition of $\rstar$ is as follows:
\begin{definition}[Relation $\rstar$ for a transition network $M$]\label{def:R-star}
Given a model $M = \tuple{\States,s_0,\Agt,\Act,do,Obs,obs}$ and sets of High players \High and Low players \Low, we define the relation $\rstar \subseteq \States \times \States$ as the least fixpoint of the following function $F$, transforming relations on $\States$:
\begin{eqnarray*}
F(R) &=& R_0\ \cup\\
&&\{(t_1,t_2) \mid \exists (s_1,s_2)\in R, l\in \Low, a\in\Act.
  do(s_1, l,a) = t_1, do(s_2,l,a) = t_2\} \cup\\
&&\{(t_1,t_2) \mid \exists s\in \States. (t_1,s)\in R \& (s,t_2)\in R\},\\
\end{eqnarray*}

\vspace{-0.7cm}
\noindent
where $(s_1,s_2)\in R_{0}$ iff for some sequence of personalized actions of High players $\alpha$, either $s_1,\xrightarrow{\alpha}s_2$, or $s_2\xrightarrow{\alpha}s_1$.
\end{definition}

%\smallskip
%It is straightforward to see that $\rstar$ is an equivalence relation (for reflexivity, notice that for any $s\in\States$, $s\xrightarrow{\alpha} s$ for $\alpha=\langle\rangle$, and therefore $(s,s)\in R_0$).
%Thus, it partitions $\States$ into equivalence classes: $[\States]_{\rstar}=\set{[s]_{\rstar}\mid s\in \States}$.
%We will now show that \emph{if $M$ satisfies noninterference then $\rstar$ is the smallest unwinding relation}.
%Note that the construction of $\rstar$ is independent from the actual observations in $M$.
%Conversely, if $M$ does not satisfy noninterference then $\rstar$ indicates pairs of states that \emph{must bear the same observations for Low if we want to make the model non-interferent}.
%
\short{ %begin short
  % **********
  It can be shown that it is sufficient to unify Low's observations in states connected by $\rstar$ in order to obtain a non-interferent model. In consequence, $\rstar$ generates the minimal unification that achieves the task.
  % **********
}%end short
\extended{ %begin extended
  % **********
  It is straightforward to see that $\rstar$ is an equivalence relation (for reflexivity, notice that for any $s\in\States$, $s\xrightarrow{\alpha} s$ for $\alpha=\langle\rangle$, and therefore $(s,s)\in R_0$).
  %Thus, it partitions $\States$ into equivalence classes: $[\States]_{\rstar}=\set{[s]_{\rstar}\mid s\in \States}$.
  We will now show that \emph{if $M$ satisfies noninterference then $\rstar$ is the smallest unwinding relation}.
  %Note that the construction of $\rstar$ is independent from the actual observations in $M$.
  Conversely, if $M$ does not satisfy noninterference then $\rstar$ indicates pairs of states that \emph{must bear the same observations for Low if we want to make the model $M$ non-interferent}. We will later show that it is sufficient to unify Low's observations in states connected by $\rstar$ in order to obtain a non-interferent variant of $M$. In consequence, $\rstar$ generates the minimal unification that achieves the task.

  The following proposition shows that if $M$ satisfies the noninterference property, then $\rstar$ is a subset of any unwinding relations on the states of $M$.
  \begin{proposition}\label{prop:unwind-rstar}
  Given a model $M$ and sets of players \High and \Low, if $\uwNIL$ is an unwinding relation on the states of $M$ and relation $\rstar$ is defined as in Definition \ref{def:R-star}, then $\rstar\subseteq \uwNIL$.
  \end{proposition}
  \begin{proof}
  See the Appendix.
  \end{proof}

  Now we show that if the model $M$, $\Low$ players have the same observations at any two states related by $\rstar$, then $M$ satisfies noninterference.
  \begin{lemma}\label{prop:rstarNI}
  In a model $M$ with sets of players \High and \Low, if for any $l\in \Low$, $s_1,s_2\in \States$ it is the case that $(s_1,s_2)\in \rstar$ implies $obs(s_1,l)=obs(s_2,l)$, then $R^*$ is an unwinding relation on the states of $M$ and therefore it holds that $NI_M(\High,\Low)$.
  \end{lemma}

  \begin{proof}
  We prove this by showing that $\rstar$ satisfies the conditions of Definition \ref{def:unwind-nonint}: The relation $\rstar$ is an equivalence relation, condition (OC) follows from the assumption of this lemma, and conditions (SC) and (LR) follow from the definition of the relation $\rstar$. Therefore it holds that $NI_M(\High,\Low)$.
  \end{proof}

  And as the last step before introducing the unification of function $\ustar$, we show that if $M$ satisfies noninterference, then $\rstar$ is an unwinding relation on its states (and by Proposition~\ref{prop:unwind-rstar} it is in fact the smallest unwinding relation).
  \begin{proposition}\label{prop:rstar-is-unwinding-for-NI}
  In a model $M$ with sets of players \High and \Low, if $NI_M(\High,\Low)$ then $\rstar$ is an unwinding relation on states of $M$.
  \end{proposition}
  \begin{proof}
  See the Appendix.
  \end{proof}
  % **********
}%end extended
Now, by using relation $\rstar$, we define the unification of observations $\ustar$ that will provide the noninterferent idealized variant of $M$.

\begin{definition}[Unification for noninterference $\ustar$]\label{def:U-star}
We define the unification of observations $\ustar \subseteq Obs\times Obs$ as follows.
For any $o_1,o_2\in Obs$, we have $(o_1,o_2)\in \ustar$ iff there exist $s_1,s_2,t_1,t_2\in \States$ and $l\in\Low$ such that:
\short{
  % **********
  \\
  (a) $obs(s_1,l)=o_1$,\
  (b) $obs(s_2,l)=o_2$,\
  (c) $(s_1,t_1)\in \rstar$,\
  (d) $(s_2,t_2)\in \rstar$, and\
  (e) $obs(t_1,l)=obs(t_2,l)$.
  % **********
}%end short
\extended{
  % **********
  \begin{description}
  \item{(a)} $obs(s_1,l)=o_1$,
  \item{(b)} $obs(s_2,l)=o_2$,
  \item{(c)} $(s_1,t_1)\in \rstar$,
  \item{(d)} $(s_2,t_2)\in \rstar$, and
  \item{(e)} $obs(t_1,l)=obs(t_2,l)$.
  \end{description}
  % **********
}%end extended
\end{definition}

%Relation $\ustar$ defined this way is in fact (as we show in the following) the minimal equivalence relation satisfying the condition that if $\rstar(s_1,s_2)$ then $\ustar(obs(s_1),obs(s_2))$.

It then holds that $\ustar(M)$ satisfies the noninterference property (Proposition~\ref{prop:rstar-is-unwinding-for-ustar-unified-model}) and no refinement of $\ustar$ achieves that (Proposition~\ref{prop:ustar-is-minimal-NI-unification}).
%\short{
%\footnote{
%  Again, we omit the proofs in this version of the paper due to lack of space. }
%}% end-short
%
\extended{%begin extended
  % **********
  The following lemma states that if two states are related by $\rstar$, then their observations are unified by $\ustar$.
  \begin{lemma}\label{lemma:rstar-implies-ustar}
  In a model $M$, for any $s_1,s_2\in\States$ and $l\in\Low$, if $(s_1,s_2)\in \rstar$ then $(obs(s_1,l),obs(s_2,l))\in \ustar$.
  \end{lemma}
  \begin{proof}
  See the Appendix.
  \end{proof}

  As the next step, we show that $\ustar(M)$ satisfies the noninterference property.

  % **********
}%end extended
\begin{proposition}\label{prop:rstar-is-unwinding-for-ustar-unified-model}
Given a model $M$, and $\ustar(M) = \tuple{\States,s_0,\Agt,\Act,do,Obs^*,obs^*}$ defined as in Definition \ref{def:U-star} on $M$, it holds that $NI_{\ustar(M)}(\High,\Low)$.
%
%it holds that $\rstar$ is an unwinding relation for $\ustar(M')$ and therefore $NI_{\ustar(M')}(\High,\Low)$.
\end{proposition}
\extended{%begin extended
  % **********
  \begin{proof}
  See the Appendix.
  \end{proof}

  \begin{figure*}[h]\centering
  \hspace*{-30pt}
  \begin{tikzpicture}[thick,scale=0.6, every node/.style={transform shape}]
  \input{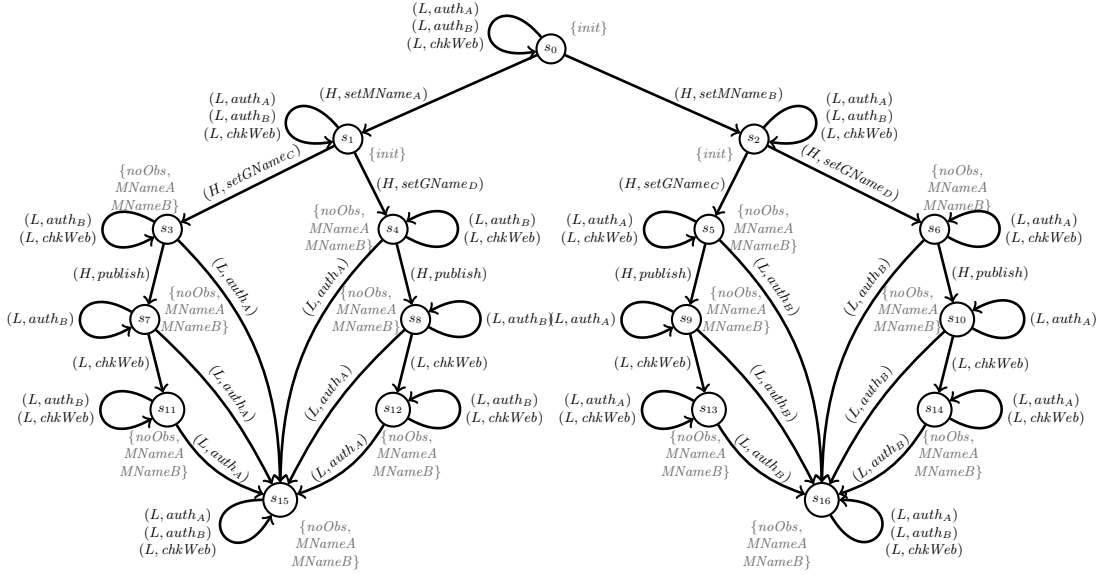}
  \end{tikzpicture}
  \caption{The noninterfering idealized variant of the banking model $M_b$}
  \label{fig:banking-idealized}
  \end{figure*}

  As the last step before proving Theorem~\ref{prop:uniqueness}, we show that $\ustar$ is the minimal unification that makes the model $M$ noninterfering.
  % **********
}%end extended
\begin{proposition}\label{prop:ustar-is-minimal-NI-unification}
Given a model $M$, and sets of players \High and \Low, for any unification of observations $U$ where $U(M) = \tuple{\States,s_0,\Agt,\Act,do,Obs',obs'}$, if $NI_{U(M)}(\High,\Low)$ then $\ustar\subseteq U$.
\end{proposition}

%\begin{proof}
%See the Appendix.
%\end{proof}

\smallskip
We can now complete the proof of Theorem~\ref{prop:uniqueness}.

\begin{proof}[of Theorem~\ref{prop:uniqueness}]
We want to prove that, given a model $M$, set of players \High and \Low, and any unification of observations $\mathcal{U}$, if $\mathcal{U}(M)$ is a noninterfering idealized variant of $M$, then $\mathcal{U}=\ustar$.
Assume that $\mathcal{U}(M)$ is a noninterfering idealized variant of $M$. By property (i) of Definition \ref{def:NI-idealized} and Proposition \ref{prop:ustar-is-minimal-NI-unification} we infer that $\ustar\subseteq \mathcal{U}$. Also, by Proposition \ref{prop:rstar-is-unwinding-for-ustar-unified-model}, we have that $NI_{\ustar(M)}(\High,\Low)$. Therefore by property (ii) of Definition \ref{def:NI-idealized} it holds that $\mathcal{U}=\ustar$.
\end{proof}

%As $\ustar$ is uniquely defined for a given model $M$, and any noninterfering idealized variant of $M$ is equal to $\ustar$, it entails the uniqueness of the noninterfering idealized  variant of a given model $M$. Moreover,

\extended{
  % **********
  From now on, we assume that $\Ideal(M)$ refers to the noninterfering idealized variant of $M$.
  % **********
}%end extended

\begin{example}\label{ex:unification}
Consider models $M_a, M_b$ in Figure~\ref{fig:banking2a} and Figure~\ref{fig:banking2a}. We recall that both models are {not} noninterferent.
In the noninterferent idealized variant of $M_a$, observations $\obsL{noObs}$, $\obsL{MnameC}$, and $\obsL{MNameD}$ of \Low are unified and replaced by the equivalence class $\set{\obsL{noObs},\obsL{MNameD},\obsL{MNameD}}$.
The idealized variant of $M_b$ is constructed analogously by unification of $\obsL{noObs}$, $\obsL{MnameA}$, and $\obsL{MNameB}$.
\extended{Figure~\ref{fig:banking-idealized} shows the idealized variant $Ideal(M_b)$ of $M_b$.}
Clearly, \Low has no surely winning strategy to guarantee $\overline{\Gamma} = \Gamma_\reachst$ for $\reachst = \set{s_{15},s_{16},\sherr}$ in both $\Ideal(M_a)$ and $\Ideal(M_b)$.

Recall from Example \ref{exmpl:effective-secure} that \Low has no winning strategy for $\overline{\Gamma}$ in $M_a$, but she has one in $M_b$. So, $M_a \simeq_{\Low,\Gamma} \Ideal(M_a)$, but $M_b \not\simeq_{\Low,\Gamma} \Ideal(M_b)$. Thus, $M_a$ is effectively information-secure for $(\Low,\Gamma)$, but $M_b$ is not.
\qed
\end{example}

It is important to notice that noninterferent variants are indeed idealizations:

\begin{proposition}
For every $M$, \Low, and $\Gamma$, we have that $M \preceq_{\Low,\Gamma} \Ideal(M)$.
\end{proposition}
\begin{proof}
Note that because $M$ and $\Ideal(M)$ differ only in their observation functions. Also we have that for any pair of states $s_1,s_2\in\States$, if $[s_1]_\Low^M = [s_2]_\Low^M$ then $[s_1]_\Low^{\Ideal(M)} = [s_2]_\Low^{\Ideal(M)}$. Therefore all the strategies of \Low in $\Ideal(M)$ are also \Low's strategies in $M$. Thus for any for any goal $\Gamma \subseteq paths(M)$, if \Low have a surely winning strategy to enforce $\overline{\Gamma}$ in $\Ideal(M)$ then they also have a surely winning strategy for $\overline{\Gamma}$ in $M$, \qed.
\end{proof}

\extended{ %begin extended
  % **********
  \medskip
  Finally, note that the concept of noninterference in our construction of effective security can be in principle replaced by an arbitrary constraint of information leakage.
  The same reasoning scheme could be applied to noninference, nondeducibility, strategic noninterference, and so on. The pattern does not change: given a ``classical'' property $\mathcal{P}$ of information security, we define the idealized variant of $M$ through the minimal unification $U$ such that that $U(M)$ satisfies $\mathcal{P}$. Then, $M$ is effectively secure in the context of property $\mathcal{P}$ iff it is strategically equivalent to $U(M)$.

  We leave the investigation of which information security properties have unique minimal unifications for future work.
  % **********
}%end extended

\section{Extending the Results to a Broader Class of Models}\label{sec:non-total-models}

As mentioned before, the models of Goguen and Meseguer are ``total on input,'' i.e., each action label is available to every user at every state. This makes modeling actual systems very cumbersome.
\extended{
  % **********
  We have seen that in the previous examples where spurious states had to be added to the analysis to allow for some synchronization between actions of different agents.
  % **********
}%end extended
In this section, we consider a broader class of models, and show how our results carry over to the more expressive setting. That is, we consider \emph{partial transition networks (PTS)} $M = \tuple{\States,s_0,\Agt,\Act,Obs,obs,do}$ which are defined as in Section~\ref{sec:models}, except that the transition function $do : \States \times \Agt \times \Act \fpart \States$  can be a partial function. By $do(s,u,a)=\undef$ we denote that action $a$ is unavailable to user $u$ in state $s$; additionally, we define $\act(s,u)=\set{a\in\Act \mid do(s,u,a)\neq\undef}$ as the set of actions available to $u$ in $s$.
Moreover, we assume that players are aware of their available actions, and hence can distinguish states with different repertoires of choices -- formally, for any $u\in\Agt, s_1, s_2 \in\States$, if $obs(s_1,u)=obs(s_2,u)$ then $\act(s_1,u)=\act(s_2,u)$.

We begin by a suitable update of the definition of noninterference:

\begin{definition}[Noninterference for partial transition networks]\label{def:noninterference-nontotal}
Given a PTS $M$ and sets of agents $\High,\Low$, such that $\High\cup\Low=\Agt, \High\cap\Low=\emptyset$, we say that \emph{$\High$ is non-interfering with $\Low$} iff for all $\alpha\in(\Agt\times\Act)^*$ and all $\low\in \Low$, if $exec(\alpha)\neq \undef$ then $[exec(\alpha)]_{\low} = [exec(\Purge_\High(\alpha))]_{\low}$. We denote the property also by $\NI_M(\High,\Low)$, thus slightly overloading the notation.
%\footnote{  The original definition of noninterference from~\cite{goguen1982security} omits the condition ``if $exec(\alpha)\neq \undef$''. Clearly, both definitions agree in the class of models considered in~\cite{goguen1982security}, i.e., models that are total on input. This is because, in those models, all sequences of personalized actions can be executed from any state in the system. }
\end{definition}
Note that Definition~\ref{def:noninterference} is a special case of Definition~\ref{def:noninterference-nontotal}.
We now define the noninterferent idealized variant based on the \emph{total extension} of a PTS.

\begin{definition}[U-total extension]
Given a PTS $M = \tuple{\States,s_0,\Agt,\Act,Obs,obs,do}$ and a subset of users $U\subseteq\Agt$, we define the $U$-total variant of $M$ as $total_{U}(M)=\tuple{\States,s_0,\Agt,\Act,Obs,obs,do'}$ where the transition function $do'(.)$ is defined as follows: for every $s\in\States$, $v\in\Agt$ and $a\in\Act$, $do'(s,v,a)=s$ if for some $u\in U$ we have $v=u$ and $do(s,u,a)=\undef$, otherwise $do'(s,v,a)=do(s,v,a)$.
\end{definition}

\begin{definition}[Noninterferent idealized model for PTN]\label{def:NI-idealized-nontotal}
Given a partial transition network $M$ and a set of ``low'' players \Low, we define the \emph{noninterferent idealized variant of $M$} as $\mathcal{U}(total_L(M))$ such that:
\begin{enumerate}
\item[(i)] $\NI_{\mathcal{U}(total_L(M))}(\High,\Low)$, and
\item[(ii)] for every $\mathcal{U}' \subsetneq \mathcal{U}$ it is not the case that $\NI_{\mathcal{U}'(total_L(M))}(\High,\Low)$.
\end{enumerate}
%
%We will denote $\Ideal^{\NIblind}_L(M)=\set{M'\in tEQ(M)\mid \NI_{M'}(\Agt\setminus\Low,\Low,S)}$.
\end{definition}

\extended{The uniqueness theorem is then stated similar to Theorem~\ref{prop:uniqueness}:}

\begin{theorem}\label{prop:uniqueness-nontotal}
For every partial transition network $M$, there is always a unique unification $\mathcal{U}$ satisfying properties (i) and (ii) from Definition~\ref{def:NI-idealized-nontotal}.
\end{theorem}

The proof is similar to the proof of Theorem~\ref{prop:uniqueness}, with the difference that we use $R^*_{total_L(M)}$ instead of $\rstar$ for constructing the idealized variant. However, as we use the concept of unwiding relation as the basis for using the $R^*$ relation for constructing the idealized variant, we first need to modify the definition of the unwinding relation in Definition~\ref{def:unwind-nonint} and its corresponding proposition, Proposition~\ref{prop:unwind-nonint} to adapt them to the new model:

\begin{definition}[Unwinding for Noninterference in PTN]\label{def:unwind-nonint-nontotal}
\extended{Let $M$ be a transition network, $\High$ a set of High agents, and $\Low$ a set of Low agents. Then,}
$\uwNIL \subseteq \States \times \States$ is an \emph{unwinding relation} iff it is an equivalence relation satisfying the conditions of \emph{output consistency (OC)}, \emph{step consistency (SC)}, and \emph{local respect (LR)}. That is, for all states $s,t \in \States$:
\begin{description}
\item[(OC)] If $s\uwNIL t$ then $[s]_\Low=[t]_\Low$;

\item[(SC)] If $s\uwNIL t$, $u\in \Low$, and $a\in \Act$ then $a\in act(s,u)$ implies $do(s,u,a)\uwNIL do(t,u,a)$;

\item[(LR)] If $u\in \High \textrm{ and } a\in \Act$ then $a\in act(s,u)$ implies $s \uwNIL do(s,u,a)$.
\end{description}
\end{definition}

\begin{proposition}\label{prop:unwind-nonint-nontotal}
%In a transition network $M$, with set of High agents $\High$ and a set of Low agents $\Low$, we have
$NI_M(\High,\Low)$ iff there exist an unwinding relation $\uwNIL$ on the states of $M$ that satisfies (OC), (SC) and (LR).
\end{proposition}

The rest of the proof of Theorem~\ref{prop:uniqueness-nontotal} follows analogously.

\begin{example}
With PTS, the scenario from Example~\ref{ex:ni} can be modeled directly, without spurious states that ruled out illegal transitions.
Thus, our models $M_a, M_b$ for the two variants of the scenario are now exactly depicted in Figures~\ref{fig:banking2a} and~\ref{fig:banking2b}.

The noninterferent idealized variants of $M_a$ (resp.~$M_b$) is again obtained by the unification of observations $\obsL{noObs}$, $\obsL{MnameC}$, and $\obsL{MNameD}$ (resp.~$\obsL{noObs}$, $\obsL{MnameA}$, and $\obsL{MNameB}$).
Clearly, \Low has no surely winning strategy to guarantee $\overline{\Gamma} = \Gamma_\reachst$ for $\reachst = \set{s_{15},s_{16}}$ in $M_a$, $\Ideal(M_a)$, and $\Ideal(M_b)$. Moreover, he has a surely winning strategy in $M_b$.
In consequence, $M_a$ is effectively information-secure for $(\Low,\Gamma)$, but $M_b$ is not.
\qed
\end{example}

%\para{Idealization in partial transition networks.}
The noninterferent variant was indeed an idealization in simple transition networks of Goguen and Mesguer.
Is it still the case in partial transition networks?
That is, is it always the case that \Low has no more abilities in $Ideal(M)$ than in $M$? In general, no. On one hand, \Low's observational capabilities are more limited in $Ideal(M)$, and in consequence some strategies in $M$ are no longer uniform in $Ideal(M)$. On the other hand, unification $U^*$ possibly adds new transitions to $M$, that can be used by \Low in $Ideal(M)$ to construct new strategies. However, under some reasonable assumptions, $Ideal(M)$ does provide idealization, as shown in the two propositions below.
\short{The proofs are relatively simple, and we omit them due to lack of space.}

\begin{proposition}
Let $M$ be a PTN such that for every state $s$ in $M$ there is at least one player $u\notin\Low$ with $act(s,u)\neq\emptyset$. Then, for any $\Gamma$, we have that $M \preceq_{\Low,\Gamma} Ideal(M)$.
\end{proposition}
%%
%\begin{proof}
%Straightforward from Proposition~\ref{prop:refines}.
%\end{proof}

\begin{proposition}
For any PTN $M$ and safety goal $\Gamma$, we have $M \preceq_{\Low,\Gamma} Ideal(M)$.
\end{proposition}

\section{Conclusions}\label{sec:conclusions}

In this paper, we introduce the novel concept of \emph{effective information security}. The idea is aimed at assessing the relevance of information leakage in a system, based on how much the leakage enables an adversary to harm the correct behavior of the system. This contrasts with the common approach to information flow security where revealing any information is seen as being intrinsically harmful. We say that two information flows are \emph{effectively equivalent} if the strategic ability of the adversary is similar in both of them. Moreover, one of them is \emph{less effectively secure} than the other one if the amount of information leaked to the adversary in it increases the damaging ability of the adversary.

In order to determine how critical the information leakage is in a given system, we compare the damaging ability of the adversary to his ability in the idealized variant of the model. We define idealized models based on noninterference, and show that the construction is well defined.
We prove this first for the deterministic, fully asynchronous transition networks of Goguen and Meseguer, and then extend the results to structures that allow for a more flexible modeling of interaction.
The construction includes an algorithm that computes the idealized variant of each model in polynomial time wrt the size of the model.

Note that the concept of noninterference in our construction of effective security can be in principle replaced by an arbitrary property of information flow.
The same reasoning scheme could be applied to noninference, nondeducibility, strategic noninterference, and so on. The pattern does not change: given a property $\mathcal{P}$, we define the idealized variant of $M$ through the minimal unification $U$ such that $U(M)$ satisfies $\mathcal{P}$. Then, $M$ is effectively information-secure in the context of property $\mathcal{P}$ iff it is strategically equivalent to $U(M)$.
We leave the investigation of which information security properties have unique minimal unifications for future work.
Moreover, we are currently working on a more refined version of effective information security based on coalitional effectivity functions\extended{~\cite{Abdou91effectivity}}, in which the strategic ability of the adversary is not only compared at the initial state of the system, but across the whole state space.

\bibliographystyle{plain}
%\bibliography{wojtek,wojtek-own,masoud}

%\appendices
\extended{ %begin extended
  % **********
  \section*{Appendix}
  This appendix contains the proofs of some of the propositions and lemmas in the paper.

  \begin{proof}[\textbf{Proof of Proposition \ref{prop:unwind-rstar}}]
  As the relation $\rstar$ is constructed by adding related states in several steps, we do the proof by induction on these steps. We show that firstly $R_0\subseteq \uwNIL$, and secondly all related pair of states added in each step also is in $\uwNIL$.

  \noindent\textbf{Induction base:} If $(s_1,s_2)\in R_0$, then for some sequence of personalized actions of High players $\alpha$, either $s_1,\xrightarrow{\alpha}s_2$, or $s_2\xrightarrow{\alpha}s_1$, hence by property (LR) , and transitivity of the unwinding relation it holds that $(s_1,s_2)\in\uwNIL$. Therefore $R_0\subseteq\uwNIL$.

  \noindent\textbf{Induction step:} We show that if $R_i\subseteq\uwNIL$, then $F(R_i)\subseteq\uwNIL$ holds. $F(R_i)$ is constructed by union of three sets. We show that all these three sets are subsets of $\uwNIL$:

  \noindent \textbf{i-} $R_i\in\uwNIL$ by the induction step assumption.

  \vspace{0.5pt}
  \noindent \textbf{ii-} If $(s_1,s_2)\in R_i$ then by induction step assumption $(s_1,s_2)\in\uwNIL$ . So for any $t_1,t_2\in\States$, $a\in\Act$ and $l\in \Low$  such that $do'(s_1,l,a)=t_1$ and $do'(s_2,l,a)=t_2$, by property (SC) of unwinding relation it holds that $(t_1,t_2)\in\uwNIL$. Therefore:
  \begin{align*}
  &\{(t_1,t_2) \mid \exists (s_1,s_2)\in R, l\in \Low, a\in\Act \cdot\\
  &\qquad do'(s_1, l,a) = t_1, do'(s_2,l,a) = t_2\}\\
  &\subseteq\uwNIL.
  \end{align*}
  \vspace{0.5pt}
  \noindent \textbf{iii-} If $(t_1,s)\in R_i$ and $(s,t_2)\in R_i$, then by induction step assumption it holds that $(t_1,s)\in\uwNIL$ and $(t_2,s)\in\uwNIL$ . Therefore by transitivity of $\uwNIL$ it entails that $(t_1,t_2)\in\uwNIL$, hence:
  \begin{align*}
  &\set{(t_1,t_2) \mid \exists s\in \States\cdot(t_1,s)\in R \text{ and } (s,t_2)\in R}\\
  &\subseteq\uwNIL.
  \end{align*}

  \vspace{5pt}
  By i, ii, and iii we infer that $F(R_i)\subseteq \uwNIL$, and therefore by induction base and induction step we have that $\rstar\subseteq \uwNIL$.
  \end{proof}

  \vspace{5pt}

  \begin{proof}[\textbf{Proof of Proposition \ref{prop:rstar-is-unwinding-for-NI}}]
  If $NI_{M}(\High,\Low)$ then by Proposition \ref{prop:unwind-nonint} there is an unwinding relation $\uwNIL$ on the states of $M$. By Proposition \ref{prop:unwind-rstar} $\rstar\subseteq\uwNIL$ and therefore for any $l\in \Low$, $s_1,s_2\in \States$ such that $(s_1,s_2)\in\rstar$ it is the case that $(s_1,s_2)\in\uwNIL$ and therefore $obs(s_1,l)=obs(s_2,l)$. Hence by Lemma \ref{prop:rstarNI} $\rstar$ is an unwinding relation on the states of $M$.
  \end{proof}

  \vspace{5pt}

  \begin{proof}[\textbf{Proof of Lemma \ref{lemma:rstar-implies-ustar}}]
  Assume $(s_1,s_2)\in \rstar$ and $l\in\Low$. For proving that $(obs(s_1,l), obs(s_2,l))\in \ustar$ we verify the conditions in Definition \ref{def:U-star}. By taking $obs(s_1,l)=obs_1$ and $obs(s_2,l)=obs_2$, conditions (a) and (b) are satisfied trivially. If we take $t_1:=s_2$ and $t_2:=s_2$, then $(s_1,t_1)\in\rstar$ by the proposition assumption and $(s_2,t_2)\in\rstar$ by reflexivity of $\rstar$. These prove conditions (c) and (d). Condition (e) is also satisfied because $t_1=t_2$. Therefore $(obs(s_1,l), obs(s_2,l))\in \ustar$.
  \end{proof}

  \vspace{5pt}

  \begin{proof}[\textbf{Proof of Proposition \ref{prop:rstar-is-unwinding-for-ustar-unified-model}}]
  For the proof, we are going to use Lemma~\ref{prop:rstarNI} and show that for any two states $s_1,s_2$ and $l\in\Low$, if $(s_1,s_2)\in R_{\ustar(M)}^*$ then $obs^*(s_1,l)=obs^*(s_2,l)$. First notice that $\rstar = R_{\ustar(M)}^*$, because $M$ and $\ustar(M)$ differ only in their observation functions and the definition of $R^*$ relation does not depend on the observation function of the model. So for any $(s_1,s_2)\in R_{\ustar(M)}^*$ and $l\in\Low$ we have that $(s_1,s_2)\in\rstar$, and by Lemma \ref{lemma:rstar-implies-ustar} it follows that $(obs(s_1,l), obs(s_2,l))\in \ustar$, and therefore $obs^*(s_1,l)=obs^*(s_2,l)$. Hence by  Lemma \ref{prop:rstarNI} it holds that $R_{\ustar(M)}^*$ is an unwinding relation for $\ustar(M)$'and therefore $NI_{\ustar(M)}(\High,\Low)$.
  \end{proof}

  \vspace{5pt}

  \begin{proof}[\textbf{Proof of Proposition \ref{prop:ustar-is-minimal-NI-unification}}]
  Assume $U$ is a unification of observations for model $M$ such that $NI_{U(M)}(\High,\Low)$ and assume $(obs_1,obs_2)\in\ustar$. We show that $(obs_1,obs_2)\in U$ and hence $\ustar\subseteq U$. By the definition of $\ustar$, there exists $s_1,s_2,t_1, t_2\in\States$, $l\in\Low$ such that $obs(s_1,l)=obs_1$, $obs(s_2,l)=obs_2$, $(s_1,t_1)\in\rstar$, $(s_2,t_2)\in\rstar$ and $obs(t_1,l)=obs(t_2,l)$. By $NI_{U(M)}(\High,\Low)$ and Proposition \ref{prop:rstar-is-unwinding-for-NI} we have that $R_{U(M)}^*$ is an unwinding relation for $U(M)$. So as $\rstar = R_{U(M)}^*$, $\rstar$ is also an unwinding relation for $U(M)$. Therefore by property (OC) of unwinding relation, from $(s_1,t_1)\in\rstar$ and $(s_2,t_2)\in\rstar$ we entail that $obs'(s_1,l)=obs'(t_1,l)$ and $obs'(s_2,l)= obs'(t_2,l)$. Using the definition of $obs'(.)$ we have that $(obs(s_1,l),obs(t_1,l))\in U$ and $(obs(s_2,l),obs(t_2,l))\in U$. So, as $obs(t_1,l)=obs(t_2,l)$ and by transitivity property of $U$, we infer that $(obs(s_1),obs(s_2))\in U$, and it follows that $(obs_1,obs_2)\in U$. Therefore $\ustar\subseteq U$.
  \end{proof}
  % **********

  \vspace{5pt}

  \begin{proof}[\textbf{Proof of Proposition \ref{prop:unwind-nonint-nontotal}}]
  \para{\underline{``$\pmb{\Leftarrow}$''}}\ Suppose that there exists an unwinding relation $\uwNIL$ satisfying (OC), (SC) and (LR). We show for all $\alpha\in(\Agt\times\Act)^*$ and $\low\in \Low$, if $exec(\alpha)\neq \undef$ then $[exec(\alpha)]_{\low} = [exec(\Purge_\High(\alpha))]_{\low}$. We prove by induction on the size of $\alpha$.

  \noindent\textbf{Induction base:} $\alpha = \langle \rangle$. In this case $\Purge_\High(\alpha) = \langle \rangle$, and therefore $exec(\alpha)=exec(\Purge_\High(\alpha))=s_0$. By the reflexivity of $\uwNIL$ we have that $exec(\alpha) \uwNIL exec(\Purge_\High(\alpha))$.

  \noindent\textbf{Induction step:} Suppose for some $\alpha\in(\Agt\times\Act)^*$, $exec(\alpha)\neq \undef$ implies $exec(\alpha)\uwNIL exec(\Purge_\High(\alpha))$. We show that for all $a\in \Act$ and $u\in\Agt$ it holds that $exec(\alpha\circ(u,a))\neq \undef$ implies $exec(\alpha\circ(u,a))\uwNIL exec(\Purge_\High(\alpha\circ(u,a)))$ (where $\circ$ denotes the concatenation operator). We consider three cases:

  \noindent i) If $exec(\alpha\circ(u,a))=\undef$ then it holds that $exec(\alpha\circ(u,a))\neq \undef$ implies $exec(\alpha\circ(u,a))\uwNIL exec(\Purge_\High(\alpha\circ(u,a)))$.

  \noindent ii) If $exec(\alpha\circ(u,a))\neq\undef$ and $u\in\Low$ then firstly notice that $exec(\Purge(\alpha\circ(u,a))\neq\undef$. Because by induction step assumption and (OC) it holds that $obs(exec(\alpha),u) = obs(exec(\Purge_\High(\alpha)),u)$ and so because $a\in act(exec(\alpha),u)$, by our model restrictions it holds that $a\in act(exec(\Purge_\High(\alpha),u))$. Therefore by $exec(\alpha)\uwNIL exec(\Purge_\High(\alpha))$ (induction step assumption), $u\in\Low$, and (SC) we have $exec(\alpha\circ(u,a))\uwNIL exec(\Purge_\High(\alpha\circ(u,a)))$.

  \noindent iii) If $exec(\alpha\circ(u,a))\neq\undef$ and $u\in\High$ then by (LR) property of $\uwNIL$, $exec(\alpha)\uwNIL exec(\alpha\circ(u,a))$. By this, induction step assumption and $\Purge_\High(\alpha) = \Purge_\High(\alpha\circ(u,a))$ we infer that $exec(\alpha\circ(u,a))\uwNIL exec(\Purge_\High(\alpha\circ(u,a)))$.

  \para{\underline{``$\pmb{\Rightarrow}$''}}\ Suppose that $NI_M(\High,\Low)$, we show there exists an unwinding relation on the states of $M$. Consider the relation $\sim$ defined as follows: for any $s,t\in\States$, $s\sim t$ if for all $\alpha\in(\Agt\times\Act)^*$ and $u_L\in\Low$, it holds that if $exec(s,\alpha)\neq\undef$ and $exec(t,\alpha)\neq\undef$, then ${[exec(s,\alpha)]}_{u_L} = {[exec(t,\alpha))]}_{u_L}$. It can easily be seen that $\sim$ is an equivalence relation, we prove that it satisfies (OC), (SC) and (LR) properties.

  \noindent\textbf{(OC)}: If $s\sim t$ and we take $\alpha = \langle \rangle$, by ${[exec(s,\alpha)]}_{u_L} = {[exec(s,\alpha)]}_{u_L}$ it holds that $[s]_{u_L} = {[t]}_{u_L}$and therefore $\sim$ satisfies (OC).

  \noindent\textbf{(SC)}: Suppose that for some $s,t\in\States$, $u\in\Low$ and $a\in\Act$ such that $s\sim t$, it holds that $do(s,u,a)\neq\undef$ and $do(s,u,a)\not\sim do(t,u,a)$. Then there exists $\alpha\in(\Agt\times\Act)^*$, $u_L\in\Low$ such that $exec(do(s,u,a),\alpha)\neq\undef$, $exec(do(t,u,a),\alpha)\neq\undef$, and ${[exec(do(s,u,a),\alpha)]}_{u_L} \neq {[exec(do(t,u,a),\alpha))]}_{u_L}$. Therefore ${[exec(s,((u,a)\circ\alpha)]}_{u_L} \neq {[exec(t,((u,a)\circ\alpha)]}_{u_L}$, which contradicts $s\sim t$.

  \noindent\textbf{(LR)}: Suppose that for some $s\in\States$, $u\in\High$ and $a\in\Act$, it holds that $do(s,u,a)\neq\undef$ and $s\not\sim do(s,u,a)$. Then there exists $\alpha\in(\Agt\times\Act)^*$, $u_L\in\Low$ such that $exec(do(s,u,a),\alpha)\neq\undef$, $exec(s,\alpha)\neq\undef$, and ${[exec(s,\alpha)]}_{u_L} \neq {[exec(do(s,u,a),\alpha))]}_{u_L}$. Because $s$ is reachable, we have that $s=exec(\beta)$ for some $\beta\in(\Agt\times\Act)^*$. Therefore ${[exec(\beta\circ\alpha)]}_{u_L} \neq {[exec(\beta\circ(u,a)\circ\alpha))]}_{u_L}$. But this is a contradiction because by $NI_M(\High,\Low)$ it holds that ${[exec(\beta\circ(u,a)\circ\alpha))]}_{u_L} = {[exec(\Purge_\High(\beta\circ(u,a)\circ\alpha)))]}_{u_L}$ and ${[exec(\beta\circ\alpha))]}_{u_L} = {[exec(\Purge_\High(\beta\circ\alpha)))]}_{u_L}$ and we have that $\Purge_\High(\beta\circ(u,a)\circ\alpha)=\Purge_\High(\beta\circ\alpha)$.
  \end{proof}
}%end extended
\end{document}